\documentclass{article}

\usepackage{hyperref}

\usepackage{bbold}
\usepackage{color}
\usepackage{extarrows}
\usepackage[bbgreekl]{mathbbol}
\usepackage{microtype}%if unwanted, comment out or use option "draft"
\usepackage{multicol}
\usepackage{multirow}
\usepackage{proof}
\usepackage{stmaryrd}
\usepackage[normalem]{ulem}
\usepackage{xspace}
\usepackage[all]{xy}\CompileMatrices
\usepackage{wrapfig}
\usepackage{subfigure}
\usepackage{prooftree}
\usepackage{multirow,bigdelim}
\usepackage{xargs}
\usepackage{cases}
\usepackage{amsthm}%proof environment
\usepackage{authblk}

\DeclareSymbolFont{bbold}{U}{bbold}{m}{n}
\DeclareSymbolFontAlphabet{\mathbbold}{bbold}

\newtheorem{example}{Example}
\newtheorem{definition}{Definition}
\newtheorem{lemma}{Lemma}
\newtheorem{proposition}{Proposition}

\newcommand{\Nat}{\ensuremath{\mathbb{N}}}

\newcommand{\Bin}{\ensuremath{\mathbb{2}}}

\newcommand{\down}[1]{\ensuremath{\lfloor #1 \rfloor}}

\newcommand{\minp}[1]{\ensuremath{{}^\circ{#1}}}
\newcommand{\maxp}[1]{\ensuremath{#1^\circ}}

\newcommand{\preS}[1]{\ensuremath{{}^\bullet{#1}}}
\newcommand{\postS}[1]{\ensuremath{#1^\bullet}}

\newcommand{\BC}[1]{\textsc{bc}(#1)}
\newcommand{\bc}{\mathbb{C}}

\newcommand{\scelleq}{\leftrightarrow}

\newcommand{\co}[1]{\overline{#1}}

\definecolor{mygray}{RGB}{120,120,120}

\newcommand{\drawplace}{*[o]=<.9pc,.9pc>{\ }\drop\cir{}}
\newcommand{\drawmarkedplace}{*[o]=<.9pc,.9pc>{\bullet}\drop\cir{}}

\newcommand{\nameplaceright}[1]{\POS[]+<.9pc,0pc>\drop{\scriptstyle{#1}}}

\newcommand{\drawtrans}[1]{*=<1.5pc,.8pc>{\scriptstyle{#1}}\drop\frm{-}}

\newcommand{\Kl}[1]{\mathcal{K}\ell(#1)}
\newcommand{\D}{\mathcal{D}}

\newcommand{\reductionRule}[2]{{\prooftree{ #1}\justifies{ #2}\endprooftree}}

\newcommand{\termtype}[4]{#1 \xrightarrow{#3} #4}

\newcommandx{\bct}[3][1=m,2=o,3=\Theta]{\mathsf{C}({#3})} 
\newcommandx{\cst}{\mathsf{C}} 

\newcommandx{\cop}[2][1=m,2=T_m]{#1\triangleright #2} 

\newcommandx{\sort}[3][3=\emptyset]{\ensuremath{(#1,#2,#3)}}

\newcommand{\cseq}[1]{{\bf #1}}
\newcommandx{\tw}[1][1={a,b}]{\ensuremath{\mathsf{X}}^{#1}}
\newcommandx{\idN}[1][1=a]{{\ensuremath{\mathsf{I}}^{#1}}}

\newcommandx{\toT}[2][1=N,2=m]{\llparenthesis{#1,#2}\rrparenthesis}

\newcommand\twarr[2]{%
\mathrel{\mathop{\moverlay{\scriptstyle\xrightarrow{\,#1\,}\cr{\lower.2em\hbox{$\scriptstyle{}_{#2}$}}}}}}
\newcommand\twarrw[2]{%
\mathrel{\mathop{\moverlay{\scriptstyle\Longrightarrow\cr{\lower-.6em\hbox{$\scriptstyle{}_{#1}$}}
\cr{\lower.3em\hbox{$\scriptstyle{}_{#2}$}}}}}}

\newcommandx{\enc}[4][1=T,2={\cseq i},3={\cseq o},4=\delta]{\llbracket#1,#2,#3,#4\rrbracket}

\newcommand{\removedp}[1]{{#1}}

\newcommand{\conf}[1]{{\sf Conf}(#1)}

\newcommandx{\termsof}[1]{\llparenthesis{#1}\rrparenthesis}

\newcommand{\kleisli}[4]{\llbracket #1,#2\rrbracket^{#3}_{#4}}

\newcommand{\ket}[2]{#1 | #2\rangle}

\begin{document}

\title{Unifying Inference for Bayesian and Petri Nets\thanks{The second author has been  partially supported by CONICET grant PIP 11220130100148CO.
The third author carried on part of the work while attending a Program on Logical
Structures in Computation at Simons Institute, Berkeley, 2016.}}        
\author[1]{Roberto Bruni}
\author[2]{Hern\'an Melgratti}
\author[1]{Ugo Montanari}
\affil[1]{Dipartimento di Informatica,  Universit\`a di Pisa, Italy}
\affil[2]{Departamento de Computaci\'on, Universidad de Buenos Aires - Conicet, Argentina}

\date{}
\maketitle

% !TEX root =  main.tex

\begin{abstract}
Recent work by the authors equips Petri occurrence nets (PN) with probability distributions which fully replace nondeterminism. To avoid the so-called confusion problem, the construction imposes additional causal dependencies which restrict choices within certain subnets called structural branching cells (s-cells). 
Bayesian nets (BN) are usually structured as partial orders where nodes define conditional probability distributions. In the paper, we unify the two structures in terms of Symmetric Monoidal Categories (SMC), so that we can apply to PN ordinary analysis techniques developed for BN. Interestingly, it turns out that PN which cannot be SMC-decomposed are exactly s-cells. This result confirms the importance for Petri nets of both SMC and s-cells.
\end{abstract}

% !TEX root =  main.tex

\section{Introduction}

At first sight, Bayesian nets (BN) and Petri Nets (PN) have very different purposes: efficient/intelligent analysis of probabilistic distributions for BN, a concurrent, nondeterministic model of computation for PN. But in fact BN and PN share a similar structure: 
a partial ordering representing incremental, local evolutions via concurrent firings for
PN, the introduction of new variables with independent, conditional probabilities for BN.

A closer comparison can be carried on when equipping also PN with a suitable probability structure. A recent approach~\cite{DBLP:journals/iandc/AbbesB06,BMM18} aims at fully replacing nondeterministic choices with probability distributions, while keeping concurrency expressiveness as much as possible. The problem here is the so-called {\em confusion}: in PN with confusion, a concurrent computation may exhibit non stable decision steps: delaying a choice may change the available options, due to the action of a concurrent transition. 

The simplest example of confusion is the Petri net in Fig.~\ref{fig:simplePN}.
Transitions $a$ and $b$ are enabled but in conflict, because they compete for the token in place $\removedp1$; transition $c$ is also enabled and concurrent w.r.t. $a$ and $b$; however the firing of transition $a$ enables the transition $d$ that is in conflict with $c$. 
As a consequence, the concurrent run where $a$ and $c$ are executed puts in the same equivalence class two quite different traces, where different decisions are taken: (1) if $a$ is executed first, then two choices are taken ($a$ over $b$ and $c$ over $d$); (2) if $c$ is executed first, then only one choice is taken ($a$ over $b$).
When choices are taken according to some probability distributions, this makes it impossible to assign a unique probability to the concurrent computation with $a$ and $c$.

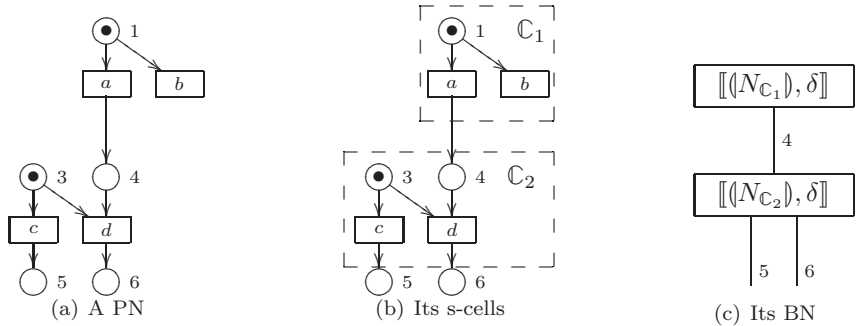
\begin{figure}
\begin{subfigure}[A PN\label{fig:simplePN}]{
          $$
          \xymatrix@R=.8pc@C=.8pc{
            &
            \drawmarkedplace\ar[d]\ar[rd]
            \nameplaceright {\removedp1}
            \\
            & \drawtrans a\ar[dd] 
            & \drawtrans b
            \\ \\
           \drawmarkedplace\ar[d]\ar[dr]
            \nameplaceright {\removedp3}
            &
            \drawplace\ar[d]
            \nameplaceright {\removedp4}
           \\
            \drawtrans c\ar[d] 
            &
            \drawtrans d\ar[d] 
            \\
            \drawplace
            \nameplaceright {\removedp5}
            &
            \drawplace
            \nameplaceright {\removedp6}
          }
          $$}
\end{subfigure}
        \qquad\qquad
\begin{subfigure}[Its s-cells\label{fig:simplecells}]{
          $$
          \xymatrix@R=.8pc@C=.8pc{
            &
            \drawmarkedplace\ar[d]\ar[rd]
            \nameplaceright {\removedp1}
            \POS[]+<1.1pc,-1pc> *+=<4.2pc,3.6pc>[F--]{}
            \POS[]+<2.5pc,0pc>\drop{\bc_1} 
            \\
            & \drawtrans a\ar[dd] 
            & \drawtrans b
            \\ \\
           \drawmarkedplace\ar[d]\ar[dr]
            \nameplaceright {\removedp3}
            \POS[]+<2.2pc,-1pc> *+=<6.6pc,3.6pc>[F--]{}
            \POS[]+<4.5pc,0pc>\drop{\bc_2} 
            &
            \drawplace\ar[d]
            \nameplaceright {\removedp4}
           \\
            \drawtrans c\ar[d] 
            &
            \drawtrans d\ar[d] 
            \\
            \drawplace
            \nameplaceright {\removedp5}
            &
            \drawplace
            \nameplaceright {\removedp6}
          }
          $$}
\end{subfigure}
        \qquad\qquad
\begin{subfigure}[Its BN\label{fig:simpleBN}]{
          $$
          \xymatrix@R=.8pc@C=.8pc{
          \\
          {\kleisli{\termsof{N_{\bc_1}}}{\delta}{}{}} \ar@{-}[dd]^{4} 
          \POS[]+<0pc,0pc> *+=<5pc,1.3pc>[F-]{}
          \\ \\
          {\kleisli{\termsof{N_{\bc_2}}}{\delta}{}{}}
          \ar@{-}^(.8){5}@<-2ex>[dd]\ar@{-}^(.8){6}@<2ex>[dd]
          \POS[]+<0pc,0pc> *+=<5pc,1.3pc>[F-]{}
          \\ \\
          {\ }
          }
          $$}
\end{subfigure}
\caption{A PN with confusion}\label{fig:simple}
\end{figure}

The solution proposed by the authors in~\cite{BMM18} is to translate the given PN into an equivalent confusionless net (ClPN). This is done by partitioning the net in \emph{structural branching cells} (\emph{s-cells}) where decisions must be resolved. S-cells are the equivalence classes of a preorder $\sqsubseteq$, that introduces some further causal dependencies. The preorder is obtained by closing transitively the relation including prime mutual exclusion and immediate causality. It follows that the preorder induces a partial order on s-cells, still denoted $\sqsubseteq$. In the example above there are two s-cells $\bc_1 \sqsubseteq \bc_2$, meaning that the choice between $a$ and $b$ must be resolved before the one between $c$ and $d$ (see Fig.~\ref{fig:simplecells}). S-cells can then be translated to a confusionless net, where the dependencies between s-cells are implemented by additional places in a way that corresponds to the execution strategy of~\cite{DBLP:journals/iandc/AbbesB06}.

To make confusionless a PN with confusion, it is necessary to delay non stable decisions until any two enabled transitions either do not share any precondition or they share all of them. Then such choice steps are equipped with probability distributions. In practice, our construction introduces a negation place $\co{p}$ for every place $p$ of the original net, and adds suitable controls to make sure that whenever place $\co{p}$  becomes inhabited, place $p$ is guaranteed never to become occupied. Thus when the present marking includes $\co{p}$, all transitions requiring $p$ can be erased and the net simplified. The process is hierarchical, because each s-cell can be further decomposed in smaller s-cells under the assumption that some place $\co{p}$  becomes inhabited.

The aim of this paper is to show that the partial order of s-cells induces a BN structure. The potential is to develop the countless applications of BN for inference and learning in the context of an expressive model like PN.
We propose a strong formal connection between PN and BN via Symmetric Monoidal Categories (SMC). 

On the side of BN, convenient categorical presentations have been recently proposed~\cite{DBLP:journals/entcs/JacobsZ16,DBLP:journals/corr/abs-1709-00322,DBLP:conf/fossacs/ClercDDG17} which, in the discrete model, represent BN as string diagrams of a SMC $\Kl\D$. Here, objects are natural numbers $n$ which express that $2^n$ cases are possible, and arrows are rectangular matrices, where rows assign probability distributions on the output cases for every input case. An arrow $f: X\rightarrow \D(Y)$ models a conditional probability distribution $P(Y | X)$. Concurrent arrows of string diagrams represent independent probability distributions. Usual inference analysis of BN, like forward and backward inference, bayesian inversion and disintegration can be made explicit as standard categorical constructions~\cite{DBLP:journals/corr/abs-1709-00322}.

A ClPN, and thus a PN, can also be mapped to an arrow of  $\Kl\D$, amenable to the same inference analysis techniques developed for BN. As for our translation PN-ClPN, this mapping is defined by well founded recursion on hierarchical branching cells. Here the effect of positive-negative information $p$/$\co{p}$ is played by associating object $1$ to a place (that is $2^1=2$ cases), which represents explicitly the two options.

Translating a ClPN into a BN is more difficult. In fact, an s-cell may produce several nodes of the BN, since the presence of negative information may break down the cell into a full BN. Thus while in $\Kl\D$ associativity of sequential composition takes care of the nested structure, in BN it will be necessary to introduce a nested version of BN, which, as far as we know, has not been proposed in the literature.

In Fig.~\ref{fig:simpleBN} we show the BN derived from the PN in Fig.~\ref{fig:simplePN}, represented as a string diagram.
There, $N_{\bc}$ is the subnet associated with the s-cell $\bc$ and $\delta$ is the family of probability distributions that rule the choices within $\bc_1$ (between $a$ and $b$) and $\bc_2$ (between $c$ and $d$ when place $\removedp4$ is marked, the trivial choice of $c$ when $\removedp4$ remains empty, i.e., they are conditional probabilitities depending on the presence/absence of tokens in $\removedp4$).
Roughly, there is one node for each s-cell and wires are associated with places.
The first node represents a variable that may take values $4$/$\co{4}$, i.e., it is the arrow
$$
\begin{array}{|c|c|c|}
\hline
& \emptyset & \{4\}
\\
\hline
\emptyset & p_b & p_a 
\\
\hline
\end{array}
: 0 \to 1
$$

\noindent
where the probabilities $p_a$ and $p_b = 1 - p_a$ are of course determined by $\delta$.
The second node represents a variable that may take all combination of values $5$/$\co{5}$ and $6$/$\co{6}$, conditioned to the value of the first variable, i.e., it is the arrow
$$
\begin{array}{|c|c|c|c|c|}
\hline
& \emptyset & \{5\} & \{6\} & \{5,6\}
\\
\hline
\emptyset & 0 & 1 & 0 & 0 
\\
\hline
\{4\} & 0 & p_c & p_d & 0
\\
\hline
\end{array}
: 1 \to 2
$$

\noindent
where, again, the values $p_c$ and $p_d = 1 - p_c$ are drawn by $\delta$.
For instance, $p_c$ is the conditional probability that the place $\removedp5$ is marked given that the place $\removedp4$ is marked.

To define the arrow in $\Kl\D$ that corresponds to a PN we exploit the monoidal category structure of nets and $\Kl\D$: first each $N$ net is uniquely decomposed in a term $\termsof{N}$ of an algebra whose constants are no further hierarchically decomposable s-cells, then the homomorphism $\kleisli{\termsof{N}}{\delta}{}{}$ returns the arrows in $\Kl\D$.

It is interesting to compare the ClPN and the $\Kl\D$ arrow for the same PN. The former model is much more informative in terms of concurrency and causality (see~\cite{DBLP:journals/corr/abs-1802-03726} for an event structure theory of persistent nets), while the latter is more straightforward in terms of structure and execution mode. It could be considered a fair algorithmic description of the execution style of~\cite{DBLP:journals/iandc/AbbesB06,BMM18} original model.

\paragraph{Structure of the paper}
In Section~\ref{sec:back} we fix the notation, recall the basics of Petri nets and occurrence nets and explain the notion of s-cell from~\cite{BMM18}.
In Section~\ref{sec:decompose} we provide a novel alternative characterisation of (the pre-oreder induced by) s-cells based on straightforward notion of parallel and sequential (de)composition of nets. This result further justifies the notion of s-cell as basic building block for occurrence nets.
In Section~\ref{sec:PNtoBN} we define the mapping from PN to BN. To this aim, an intermediate term algebra is used that builds on the decomposition defined in Section~\ref{sec:decompose} to break s-cells with non-empty initial interface into the hierarchical composition of other terms. Here some sort of case analysis is done: for each marking that can be provided to the s-cell we explore how it can be simplified (the absence of tokens allows for the removal of places and transitions). 
In Section~\ref{sec:inference} we show how the Bayesian structure can be exploited to reason about the marking of places of the original PN.
Finally, in Section~\ref{sec:conc} we draw some concluding remarks and give pointers to related and future work.

In~\ref{sec:appe} we show the correspondence between PN decomposition and the approach by Abbes and Benveniste based on event structures, which justifies the assignment of probability distributions to s-cells.

We assume the reader is familiar with some basic concepts from Bayesian networks and category theory.

% !TEX root =  main.tex

\section{Background}\label{sec:back}

%\todo[inline,caption={}]{
%Petri nets with structural branching cells.
%
%Additional causality implied by branching cells:
% \begin{itemize}
% \item  interpreter a la AB
% \item  compiler a la BMM
% \end{itemize}
% 
%Probability distribution: the same for both
%
%Bayesian networks in a monoidal category of distributions
%
%Easy representation of conditional probabilities for forward and backward inference and disintegration.
%}

\subsection{Notation}\label{sec:notation}

%We fix some notation.
We let $\Nat$ be the set of natural numbers
and $\Bin = \{0,1\}$.  
We write $U^S$ for the set of functions from $S$ to $U$:  
hence a subset of $S$ is an element of $\Bin^S$, 
and a multiset $m$ over $S$ is an element of $\Nat^S$.
A set can be seen as a multiset whose elements have unary multiplicity.  
Membership, union, difference and inclusion over sets and multisets are denoted by the
(overloaded) symbols: $\in$, $\cup$, $\setminus$ and $\subseteq$, respectively.  

Given a relation $R \subseteq S \times S$, we let 
$R^{-1}=\{(y,x) \mid (x,y)\in R\}$ be its inverse relation,
$R^+$ be its transitive closure and 
$R^*$ be its reflexive and transitive closure.  
We say that $R$ is \emph{acyclic} 
if $\forall s\in S.~ (s,s) \not\in R^+$.

\subsection{Petri Nets}\label{sec:petri}

\begin{definition}
A \emph{Petri net} $N$ is a tuple $(P,T,F)$ where: $P$ is the set of
places, $T$ is the set of transitions, and $F \subseteq (P\times
T)\cup (T\times P)$ is the \emph{flow relation}.  
\end{definition}

For $x\in P\cup T$, we denote by $\preS{x} = \{ y \mid (y,x)\in F\}$ and $\postS{x} =
\{ z \mid (x,z)\in F\}$ its \emph{pre-set} and \emph{post-set},
respectively.
We  assume that $P$ and $T$ are disjoint and non-empty
and that  $\preS{t}$ is non empty for every
$t\in T$.
We write $t: X \to Y$ for $t\in T$ with $X = \preS{t}$ and $Y =
\postS{t}$.
A \emph{marking} is a multiset $m\in\Nat^P$.
A marking denotes a state of a Petri net.
We say that the place $p\in P$ is \emph{marked} at $m$ if $p\in m$.
We write  $(N,m)$ for the net $N$ \emph{marked} by $m$.
In the following we write just $N$ for the marked net $(N,\emptyset)$.

Graphically, a Petri net is a directed bipartite graph whose
nodes are the places (circles) and transitions (rectangles) and whose arcs are the elements of $F$.
The marking $m$ is represented by inserting $m(p)$
tokens (bullets) in each place $p\in m$ (see Fig.~\ref{fig:N}).

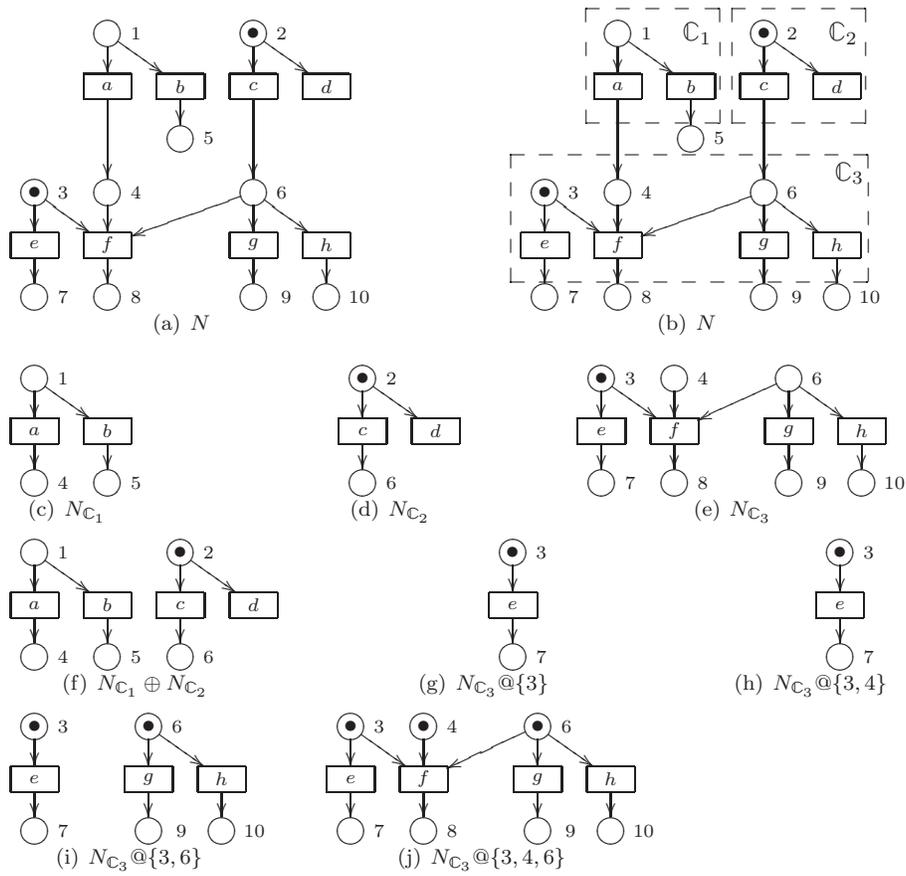
\begin{figure}[tp]
\begin{subfigure}[$N$\label{fig:N}]{
          $$
          \xymatrix@R=.8pc@C=.8pc{
            &
            \drawplace\ar[d]\ar[rd]
            \nameplaceright {\removedp1}
            & &
            \drawmarkedplace\ar[d]\ar[rd]
            \nameplaceright {\removedp2}
            \\
            & \drawtrans a\ar[dd] 
            & \drawtrans b\ar[d] 
            & \drawtrans c\ar[dd] 
            & \drawtrans d
            \\
            &
            &
            \drawplace
            \nameplaceright {\removedp5}
            &
            &
            \\
            \drawmarkedplace\ar[d]\ar[dr]
            \nameplaceright {\removedp3}
            &
            \drawplace\ar[d]
            \nameplaceright {\removedp4}
            &
            &
            \drawplace\ar[dll]\ar[d]\ar[dr]
            \nameplaceright {\removedp6}
            &
            \\
            \drawtrans e\ar[d] 
            &
            \drawtrans f\ar[d] 
            &
            &
            \drawtrans g\ar[d] 
            &
            \drawtrans h\ar[d] 
            \\
            \drawplace
            \nameplaceright {\removedp7}
            &
            \drawplace
            \nameplaceright {\removedp8}
            &
            &
            \drawplace
            \nameplaceright{\ \removedp9}
            &
            \drawplace
            \nameplaceright{\ \removedp{10}}
          }
          $$}
%\vspace{-10pt}
        \end{subfigure}
        \qquad\qquad
\begin{subfigure}[$N$\label{fig:N-bc}]{
          $$
          \xymatrix@R=.8pc@C=.8pc{
            &
            \drawplace\ar[d]\ar[rd]
            \nameplaceright {\removedp1}
            \POS[]+<1.1pc,-1pc> *+=<4.2pc,3.6pc>[F--]{}
            \POS[]+<2.5pc,0pc>\drop{\bc_1} 
            & &
            \drawmarkedplace\ar[d]\ar[rd]
            \nameplaceright {\removedp2}
            \POS[]+<1.1pc,-1pc> *+=<4.2pc,3.6pc>[F--]{}
            \POS[]+<2.5pc,0pc>\drop{\bc_2} 
            \\
            & \drawtrans a\ar[dd] 
            & \drawtrans b\ar[d] 
            & \drawtrans c\ar[dd] 
            & \drawtrans d 
            \\
            &
            &
            \drawplace
            \nameplaceright {\removedp5}
            &
             \\
            \drawmarkedplace\ar[d]\ar[dr]
            \nameplaceright {\removedp3}
            &
            \drawplace\ar[d]
            \nameplaceright {\removedp4}
            \POS[]+<2.25pc,-.8pc> *+=<11.25pc,4pc>[F--]{}
            \POS[]+<7.3pc,.6pc>\drop{\bc_3} 
            &
            &
            \drawplace\ar[dll]\ar[d]\ar[dr]
            \nameplaceright {\removedp6}
            &
            \\
            \drawtrans e\ar[d] 
            &
            \drawtrans f\ar[d] 
            &
            &
            \drawtrans g\ar[d] 
            &
            \drawtrans h\ar[d] 
            \\
            \drawplace
            \nameplaceright {\removedp7}
            &
            \drawplace
            \nameplaceright {\removedp8}
            &
            &
            \drawplace
            \nameplaceright{\ \removedp9}
            &
            \drawplace
            \nameplaceright{\ \removedp{10}}
          }
          $$}
%\vspace{-10pt}
        \end{subfigure}
        
        \begin{subfigure}[$N_{\bc_1}$\label{fig:N-bc-1}]{
          $$
          \xymatrix@R=.8pc@C=.8pc{
            \drawplace\ar[d]\ar[rd]
            \nameplaceright {\removedp1}
            \\
            \drawtrans a\ar[d] 
            & \drawtrans b\ar[d] 
            \\
            \drawplace
            \nameplaceright {\removedp4}
            &
            \drawplace
            \nameplaceright {\removedp5}
          $$}}
%\vspace{-10pt}
        \end{subfigure}
\qquad\quad
        \begin{subfigure}[$N_{\bc_2}$\label{fig:N-bc-2}]{
          $$
          \xymatrix@R=.8pc@C=.8pc{
            \drawmarkedplace\ar[d]\ar[rd]
            \nameplaceright {\removedp2}
            \\
            \drawtrans c\ar[d] 
            & \drawtrans d
            \\
            \drawplace
            \nameplaceright {\removedp6}
            &
          }
          $$}
%\vspace{-10pt}
        \end{subfigure}
        \begin{subfigure}[$N_{\bc_3}$\label{fig:N-bc-3}]{
          $$
          \xymatrix@R=.8pc@C=.8pc{
            \drawmarkedplace\ar[d]\ar[dr]
            \nameplaceright {\removedp3}
            &
            \drawplace\ar[d]
            \nameplaceright {\removedp4}
            &
            &
            \drawplace\ar[dll]\ar[d]\ar[dr]
            \nameplaceright {\removedp6}
            &
            \\
            \drawtrans e\ar[d] 
            &
            \drawtrans f\ar[d] 
            &
            &
            \drawtrans g\ar[d] 
            &
            \drawtrans h\ar[d] 
            \\
            \drawplace
            \nameplaceright {\removedp7}
            &
            \drawplace
            \nameplaceright {\removedp8}
            &
            &
            \drawplace
            \nameplaceright{\ \removedp9}
            &
            \drawplace
            \nameplaceright{\ \removedp{10}}
          }
          $$}
%\vspace{-10pt}
        \end{subfigure}
        \qquad\quad
\begin{subfigure}[$N_{\bc_1}\oplus N_{\bc_2}$\label{fig:bc1plusbc2}]{
          $$
          \xymatrix@R=.8pc@C=.8pc{
            \drawplace\ar[d]\ar[rd]
            \nameplaceright {\removedp1}
            & &
            \drawmarkedplace\ar[d]\ar[rd]
            \nameplaceright {\removedp2}
            \\
            \drawtrans a\ar[d] 
            & \drawtrans b\ar[d] 
            & \drawtrans c\ar[d] 
            & \drawtrans d
            \\
            \drawplace
            \nameplaceright {\removedp4}
            &
            \drawplace
            \nameplaceright {\removedp5}
            &
            \drawplace
            \nameplaceright {\removedp6}
            &
          }
          $$}
        \end{subfigure}
\qquad
        \begin{subfigure}[$N_{\bc_3}@\{\removedp3\}$\label{fig:N-bc-3-at-p3}]{
          $$
          \xymatrix@R=.8pc@C=.8pc{
            &&
            &\drawmarkedplace\ar[d]
            \nameplaceright {\removedp3}
            \\
            &&&
            \drawtrans e\ar[d] 
            \\
            &&&
            \drawplace
            \nameplaceright {\removedp7}
            &&
          }
          $$}
        \end{subfigure}
        \qquad\quad
        \begin{subfigure}[$N_{\bc_3}@\{\removedp3,\removedp4\}$\label{fig:N-bc-3-at-p3p4}]{
          $$
          \xymatrix@R=.8pc@C=.8pc{
            &&
            \drawmarkedplace\ar[d]
            \nameplaceright {\removedp3}
            &
             \\
            &&
            \drawtrans e\ar[d] 
            \\
            &&
            \drawplace
            \nameplaceright {\removedp7}
          }
          $$}
        \end{subfigure}
\qquad\quad
        \begin{subfigure}[$N_{\bc_3}@\{\removedp3,\removedp6\}$\label{fig:N-bc-3-at-p3p6}]{
          $$
          \xymatrix@R=.8pc@C=.8pc{
            \drawmarkedplace\ar[d]
            \nameplaceright {\removedp3}
            &
            &
            \drawmarkedplace\ar[d]\ar[dr]
            \nameplaceright {\removedp6}
            &
            \\
            \drawtrans e\ar[d] 
            &
            &
            \drawtrans g\ar[d] 
            &
            \drawtrans h\ar[d] 
            \\
            \drawplace
            \nameplaceright {\removedp7}
            &
            &
            \drawplace
            \nameplaceright{\ \removedp9}
            &
            \drawplace
            \nameplaceright{\ \removedp{10}}
          }
          $$}
        \end{subfigure}
   \quad
           \begin{subfigure}[$N_{\bc_3}@\{\removedp3,\removedp4,\removedp6\}$\label{fig:N-bc-3-at-p3p4p6}]{
          $$
          \xymatrix@R=.8pc@C=.8pc{
            \drawmarkedplace\ar[d]\ar[dr]
            \nameplaceright {\removedp3}
            &
            \drawmarkedplace\ar[d]
            \nameplaceright {\removedp4}
            &
            &
            \drawmarkedplace\ar[dll]\ar[d]\ar[dr]
            \nameplaceright {\removedp6}
            &
            \\
            \drawtrans e\ar[d] 
            &
            \drawtrans f\ar[d] 
            &
            &
            \drawtrans g\ar[d] 
            &
            \drawtrans h\ar[d] 
            \\
            \drawplace
            \nameplaceright {\removedp7}
            &
            \drawplace
            \nameplaceright {\removedp8}
            &
            &
            \drawplace
            \nameplaceright{\ \removedp9}
            &
            \drawplace
            \nameplaceright{\ \removedp{10}}
          }
          $$}
%\vspace{-10pt}
        \end{subfigure}

\caption{A simple PN}
\end{figure}

The operational semantics of a Petri net is defined by events called firings.
A transition $t$ is \emph{enabled} at the marking $m$,
written $m\xrightarrow{t}$, if $\preS{t}\subseteq m$.  
The \emph{firing} of a transition $t$ enabled at $m$ is written 
$m \xrightarrow{t} m'$ with $m' = (m\setminus \preS{t}) \cup \postS{t}$.
A firing sequence $m \xrightarrow{t_1 \cdots t_n} m'$ from $m$ to $m'$ is a finite sequence of firings, sometimes abbreviated $m \rightarrow^* m'$.
Moreover,  it is \emph{maximal} if no transition is
enabled at $m'$.
%We write $m \xrightarrow{t_1 \cdots t_n}$ if there is 
%$m'$ such that $m \xrightarrow{t_1 \cdots t_n} m'$.  
%
We say that  $m'$ is
\emph{reachable} from $m$ if  $m \rightarrow^* m'$.  
The set of markings reachable from $m$ is written $[m\rangle$. 
A marked net $(N,m)$ is \emph{safe} if each $m'\in [m\rangle$ is a set.

In the rest of the paper we only consider safe nets.
More precisely we consider so-called occurrence nets.

\subsection{Occurrence nets}\label{sec:occ}

We say that a net $(P,T,F)$ is acyclic if its flow relation $F$ is so.
Given an acyclic net we let $\preceq=F^*$ be the (reflexive) \emph{causality} relation 
and say that two transitions $t_1$ and $t_2$
are in \emph{immediate conflict}, written $t_1 \#_0 t_2$ if $t_1 \neq
t_2 \;\wedge\; \preS{t_1} \cap \preS{t_2} \neq \emptyset$.  
%Moreover, 
The \emph{conflict relation}
$\#$ is defined by letting 
$x \# y$ if there are $t_1,t_2\in T$ such that
  $(t_1,x), (t_2,y) \in F^+$ and $t_1 \#_0 t_2$.

\begin{definition}[Occurrence Net]
A \emph{nondeterministic occurrence net} (or just \emph{occurrence net}) is an acyclic net
$\mathcal{O}=(P,T,F)$ such that:
\begin{enumerate}
\item there are no backward conflicts (i.e., $\forall p\in P.~ |\preS{p}|\leq 1$), and
\item there are no self-conflicts (i.e., $\forall t\in T.~ \neg(t\# t)$).
\end{enumerate}

An occurrence net is deterministic if it does not have forward conflicts
%i.e., 
(i.e., $\forall p\in P.~|\postS{p}|\leq 1$).
\end{definition}

A place $p$ of an occurrence net $\mathcal{O}$ is called \emph{initial} if its pre-set is empty; it is called \emph{final} if its post-set is empty; it is called \emph{isolated} if it is both initial and final.
We denote by $\minp {\mathcal{O}}$ the set of its initial places and by $\maxp {\mathcal{O}}$ the set of its final places.
The net $N$ in Fig.~\ref{fig:N} is an occurrence net. The sets of its initial and final places respectively are $\minp {N} = \{\removedp1, \removedp2,\removedp3\}$
and $\maxp{N} = \{\removedp5, \removedp7,\removedp8,\removedp9,\removedp{10}\}$.

Typically it is left implicit that all the initial places of an occurrence net are marked. 
Here we need to distinguish the cases in which only some initial places are marked.

\begin{definition}[Marked Occurrence Net]
A \emph{marked occurrence net} $\mathcal{M} = (\mathcal{O},m)$ is an occurrence net $\mathcal{O}$ together with a subset $m$ of initial, non-isolated places.
\end{definition}

The idea is that:
\begin{itemize}
\item any initial place in $m$ is already marked (by one token);
\item any initial place not in $m$ can receive a token from the context.
\end{itemize}

Given a marked occurrence net $\mathcal{M} = (\mathcal{O},m)$, we denote by $\minp {\mathcal{M}}=\minp {\mathcal{O}} \setminus m$ the set of its initial (unmarked) places and by $\maxp {\mathcal{M}}= \maxp {\mathcal{O}}$ the set of its final places. For the marked occurrence net $(N,\{\removedp2, \removedp3\})$ in  Fig.~\ref{fig:N},
we have $\minp {({N},\{\removedp2, \removedp3\})}= \{\removedp1\}$ and  $\maxp {({N},\{\removedp2, \removedp3\})}= \maxp{N} = \{\removedp5, \removedp7, \removedp8,\removedp9,\removedp{10}\}$.

%\todo{Togliere? The \emph{unfolding} $\mathcal{U}(N)$ of a safe Petri net $(N,m)$
%is an occurrence net $\mathcal{O}$ together with a net morphism $\pi: \mathcal{O} \to N$ that accounts
%for all (finite and infinite) runs of $N$: its transitions 
%model all the possible instances
%of transitions in $N$ and its places 
%model all the tokens that can be created in any
%run.}

A \emph{deterministic 
nonsequential 
process} (or just \emph{process})~\cite{DBLP:journals/iandc/GoltzR83}
represents the equivalence class of all  firing sequences of a net that only differ  in the order in which concurrent firings are
executed.  
It is given as a mapping $\theta:\mathcal{D} \to N$ from a
\emph{deterministic occurrence net} $\mathcal{D}$ to $N$ (preserving pre- and post-sets).
%
%We let $\minp{\mathcal{D}} = \{ p\, \mid\, \preS p = \emptyset\}$ and
%$\maxp{\mathcal{D}} = \{ p\, \mid\, \postS p = \emptyset\}$ be the
%sets of \emph{initial} and \emph{final places} of $\mathcal{D}$,
%respectively (with $\pi(\minp{\mathcal{D}})$ be the initial marking of $N$).
%
The  firing sequences of a processes $\mathcal{D}$ are its maximal
firing sequences starting from the marking $\minp{\mathcal{D}}$.
A process of $N$ is
\emph{maximal} if its firing sequences
are maximal in $N$.

When $N$ is an acyclic safe net, the mapping
$\theta$ is just an injective graph homomorphism:
%(preserving pre- and post-sets): 
without loss of generality, we 
name the nodes in $\mathcal{D}$ as their images in $N$ and
let $\theta$ be the identity. 
%
%Abusing the notation we define $\theta$ just as the set of its underlying transitions.
%

\subsection{Structural Branching Cells}\label{sec:cells}

In~\cite{BMM18} we have proposed a solution for determining the smallest loci of decision within an acyclic finite net, called \emph{structural branching cells}: they are subnets where the decision of firing some transition is taken when it is guaranteed that no conflicting transition which is currently not enabled can become enabled in the future.

The construction in~\cite{BMM18} takes a (finite) occurrence net as input, which can be, e.g., the (truncated) unfolding of any safe net and returns a partial order of structural branching cells.

To each transition $t$ we assign a unique s-cell $[t]$.
This is achieved by taking the equivalence class of $t$ w.r.t. the equivalence relation $\scelleq$
induced by the least preorder $\sqsubseteq$ that
includes immediate conflict $\#_0$ and causality $\preceq$.
Formally, we let $\sqsubseteq$ be the transitive closure of the relation 
$\#_0\ \cup \preceq \cup\ \mathsf{Pre}^{-1}$, where  $\mathsf{Pre}= F \cap (P\times T)$.
This way, each s-cell $[t]$ also includes the places in the pre-sets of the transitions in $[t]$.
Since $\#_0$ is subsumed by the transitive closure of the relation $\preceq \cup\ \mathsf{Pre}^{-1}$,
we  equivalently set $\sqsubseteq\ = (\preceq \cup\ \mathsf{Pre}^{-1})^*$.

\begin{definition}[S-cells]\label{def:scells}
Let $N=(P,T,F)$ be a finite occurrence net and $\sqsubseteq$ defined as above.
Let $\scelleq\ = \{(x,y) \mid x\sqsubseteq y \wedge y \sqsubseteq x\}$.
The set $\BC N$ of \emph{s-cells} is 
the set of equivalence classes of $\scelleq$, i.e., $\BC N = \{[t]_{|\scelleq} \mid t\in T\}$.
\end{definition}

We let $\bc$ range over s-cells.  
It is immediate to note that s-cells are ordered by $\sqsubseteq$: we let $\bc\sqsubseteq \bc'$ if there are $t\in\bc,t'\in\bc'$ with $t \sqsubseteq t'$.
%By definition it follows that for all $\bc,\bc'\in \BC N$, if $\bc \cap \bc'\neq \emptyset$ then $\bc = \bc'$.

  For any s-cell $\bc$, we denote by $N_{\bc}$ the
  subnet of $N$ whose elements are in 
  $\bc \cup \bigcup_{t\in \bc} \postS{t}$, i.e., we include
  in $N_{\bc}$ also all places in the post-set of some transition in $\bc$.

Abusing the notation, we denote by $\minp \bc$ the set of all the
initial places in $N_{\bc}$ and by $\maxp \bc$ the set of all the
final places in $N_{\bc}$.
When the original net $(N,m)$ is marked we sometimes let its cells inherits the marking, i.e., we let the initial marking of $N_{\bc}$ be $m\cap \minp \bc$.

\begin{example} The net in Fig.~\ref{fig:N} has three s-cells, which are depicted in 
Fig.~\ref{fig:N-bc}: 
$\bc_1 = 
%[a]_{|\scelleq} = [b]_{|\scelleq} = 
\{\removedp1, a, b\}$ concerning 
the choice between $a$ and $b$,  and 
$\bc_2 = 
%[c]_{|\scelleq} = [d]_{|\scelleq} = 
\{\removedp2, c, d\}$
concerning the choice between $c$ and $d$, and 
$\bc_3 = 
%[e]_{|\scelleq} = [f]_{|\scelleq} = [g]_{|\scelleq} = [h]_{|\scelleq} = 
\{\removedp3, \removedp4, \removedp6, e, f, g, h\}$.
 The nets  $N_{\bc_1}$,  $N_{\bc_2}$ and $N_{\bc_3}$ are respectively shown 
 in Fig.~\ref{fig:N-bc-1}, \ref{fig:N-bc-2} and~\ref{fig:N-bc-3}. 
 For $\bc_1$,  $\minp{\bc_1} = \minp{N_{\bc_1}} = \{\removedp1\}$ and $\maxp{\bc_1} = \maxp{(N_{\bc_1})} =  \{\removedp4, \removedp5\}$.
  For $\bc_2$,  $\minp{\bc_2} = \minp{N_{\bc_2}}\setminus\{\removedp2\} = \{\removedp2\}\setminus\{\removedp2\} = \emptyset$ and 
  $\maxp{\bc_2} = \maxp{(N_{\bc_2})} =  \{\removedp6\}$.

\end{example}

The behaviour of a branching cell is characterised in terms of all its possible executions.

\begin{definition}[Transactions]
  Let $\bc\in \BC N$ and $m=\minp{\bc}$.  Then, a
  \emph{transaction} $\theta$ of $\bc$, written $\theta:\bc$, is a maximal (deterministic) process
  of $(N_{\bc},m)$.
  We denote by $\Theta(\bc)$
  the set of all the transactions of $\bc$.
\end{definition}

Since the set of transitions in a transaction $\theta$ uniquely
determines the corresponding process in $N_{\bc}$, we write a
transaction $\theta$ simply as the set of its transitions.
If $i=\minp{\theta}$ is the set of initial places of $\theta$ and $o=\maxp{\theta}$ is the set of its final places, we write $\theta:i\to o$.
Note that in general, for $\theta:i \to o \in  \Theta(\bc)$, we have $i\subseteq \minp{\bc}$ and $o\subseteq \maxp{\bc}$.
We write ${\sf n}(\theta)$ for the set of transitions and places of $\theta$.

\begin{example} Consider the net $N_{\bc_3}$ in Fig.~\ref{fig:N-bc-3}. It has
 the following three transactions: $\theta_1 = \{f\}$, $\theta_2 = \{e, g\}$ and $\theta_3 = \{e, h\}$,
 with 
 $\theta_1:\{\removedp3, \removedp4, \removedp6\} \to \{\removedp8\}$
 $\theta_2:\{\removedp3, \removedp6\} \to \{\removedp7, \removedp9\}$
 $\theta_3:\{\removedp3, \removedp6\} \to \{\removedp7, \removedp{10}\}$.
\end{example}

% !TEX root =  main.tex

\section{Petri Nets Decomposition}\label{sec:decompose}

We have already said that s-cells form a partial order.
Here we show that it can be seen as a particular commutative monoidal category structure. 

We proceed as follows:
\begin{enumerate}
\item we define set-theoretical parallel and sequential composition of nets;
\item we show that parallel and sequential composition, together with a suitable notion of identities, induce a commutative monoidal category structure over occurrence nets;
\item we show that s-cells are neither decomposable in parallel nor in series;
\item we show that each Petri net admits a unique maximal decomposition in terms of parallel and sequence (up to the axioms of commutative monoidal categories) and that such decomposition coincides with the partial order of s-cells.
\end{enumerate}

This provides a new characterisation of s-cells as the building blocks of occurrence nets that supports our intuition about their relevance.

Intuitively, parallel composition takes two nets and put them side by side.

\begin{definition}[Parallel composition]
Let $(P_1,T_1,F_1,m_1)$ and $(P_2,T_2,F_2,m_2)$ be two Petri nets whose nodes are disjoint (i.e., with $(P_1\cup T_1) \cap (P_2\cup T_2) = \emptyset$).
Their parallel composition is given by the element-wise union of their components:
$$(P_1,T_1,F_1,m_1) \oplus (P_2,T_2,F_2,m_2) = 
(P_1\cup P_2, T_1\cup T_2, F_1\cup F_2, m_1\cup m_2)$$
\end{definition}

Sequential composition is defined over (marked) occurrence nets only.

\begin{definition}[Sequential composition]
Let $\mathcal{M}_1=(\mathcal{O}_1,m_1)$ and $\mathcal{M}_2=(\mathcal{O}_2,m_2)$ be two marked occurrence nets, with $\mathcal{O}_j = (P_{j},T_j,F_j)$ for $j=1,2$, whose nodes are disjoint except for the final places of $\mathcal{M}_1$ that are identical to the unmarked initial places of $\mathcal{M}_2$ (i.e., with $\maxp{\mathcal{M}_1} = (P_1\cup T_1) \cap (P_2\cup T_2) = \minp{\mathcal{M}_2}$).
Their sequential composition is given by the element-wise union of their components (but note that the places in $\maxp{(\mathcal{M}_1} = \minp{\mathcal{M}_2}$ are shared):
$$(P_1,T_1,F_1,m_1) ; (P_2,T_2,F_2,m_2) = 
(P_1\cup P_2, T_1\cup T_2, F_1\cup F_2, m_1\cup m_2)$$
\end{definition}

Let us write $\mathcal{M}: i \to o$ for a marked occurrence net with $i = \minp{\mathcal{M}}$ and $o = \maxp{\mathcal{M}}$ 
Then we note that for $\mathcal{M}_j: i_j \to o_j$ for $j\in[1,4]$:
\begin{itemize}
\item $\mathcal{M}_1 \oplus \mathcal{M}_2: i_1\cup i_2 \to o_1\cup o_2$, when the parallel composition is defined;
\item $\mathcal{M}_1 ; \mathcal{M}_2: i_1 \to o_2$, when the sequential composition is defined;
\item parallel composition is commutative and associative and has
%$(\mathcal{O}_1,m_1) \oplus (\mathcal{O}_2,m_2) = (\mathcal{O}_2,m_2) \oplus (\mathcal{O}_1,m_1)$ whenever any side of the equality is defined;
%\item parallel composition is associative: $(\mathcal{O}_1,m_1) \oplus ((\mathcal{O}_2,m_2) \oplus (\mathcal{O}_3,m_3))= ((\mathcal{O}_1,m_1) \oplus (\mathcal{O}_2,m_2)) \oplus (\mathcal{O}_3,m_3)$ whenever any side of the equality is defined;
 the empty net $\mathbf{0}=(\emptyset,\emptyset,\emptyset,\emptyset):\emptyset \to \emptyset$ as neutral element, i.e. it forms a commutative monoid;
 % w.r.t. parallel composition;
%$(\mathcal{O}_1,m_1) \oplus \mathbf{0} = (\mathcal{O}_1,m_1)$;
\item sequential composition is associative;
% $(\mathcal{O}_1,m_1) ; ((\mathcal{O}_2,m_2) ; (\mathcal{O}_3,m_3))= ((\mathcal{O}_1,m_1) ; (\mathcal{O}_2,m_2)) ; (\mathcal{O}_3,m_3)$ whenever any side of the equality is defined;
\item for each set of places $i$ the identity net $I_i = (i,\emptyset,\emptyset,\emptyset): i \to i$ consisting just of (unmarked) isolated places $i$ behaves as the identity w.r.t. composition;
% $(\mathcal{O}_1,m_1) ; I_{o_1} = (\mathcal{O}_1,m_1) = I_{i_1} ; (\mathcal{O}_1,m_1)$;
\item the monoid of parallel composition is functorial: $I_\emptyset = \mathbf{0}$, $I_{i_1 \cup i_2} = I_{i_1} \oplus I_{i_2}$ and $(\mathcal{M}_1 ; \mathcal{M}_2) \oplus (\mathcal{M}_3;\mathcal{M}_4)= (\mathcal{M}_1 \oplus \mathcal{M}_3) ; (\mathcal{M}_2\oplus \mathcal{M}_4)$.
\end{itemize}

In the following, we assume $\oplus$ has higher precedence over $;$, e.g. we write $\mathcal{M}_1 \oplus \mathcal{M}_2 ; \mathcal{M}_3$ instead of $(\mathcal{M}_1 \oplus \mathcal{M}_2 ); \mathcal{M}_3$.

From the above we get that marked occurrence nets form the arrows of a commutative (strict) monoidal pre-category (it is not a monoidal category because parallel and sequential composition are defined on concrete nets and impose some disjointness requirements on their places and transitions).

\begin{example}
\label{ex:decomposition}
Consider the marked occurrence nets 
%$(N_{\bc_1},\emptyset):\{\removedp1\}\to\{\removedp4,\removedp5\}$, 
$N_{\bc_1}:\{\removedp1\}\to\{\removedp4,\removedp5\}$, 
%$(N_{\bc_2},\emptyset):\{\removedp2\}\to\{\removedp6,\}$, and  
$(N_{\bc_2},\{\removedp2\}):\emptyset\to\{\removedp6\}$, and  
$(N_{\bc_3},\{\removedp3\}): \{\removedp4,\removedp6\}\to\{\removedp7,\removedp8,\removedp9,\removedp{10}\}$ in  Fig.~\ref{fig:N-bc-1}, \ref{fig:N-bc-2} and~\ref{fig:N-bc-3}.
Note that the parallel composition of 
%$(N_{\bc_1},\emptyset)$ and $(N_{\bc_2},\emptyset)$  
$N_{\bc_1}$ and $N_{\bc_2}$  
is defined because the nets neither share places nor transitions. 
The resulting net 
%$(N_{\bc_1},\emptyset)\oplus(N_{\bc_2},\emptyset):\{\removedp1, \removedp2\}\to\{\removedp4,\removedp5, \removedp6\}$ 
$N_{\bc_1}\oplus (N_{\bc_2},\{\removedp2\}):\{\removedp1\}\to\{\removedp4,\removedp5, \removedp6\}$ 
is shown in Fig~\ref{fig:bc1plusbc2}. 
We remark that neither 
%$(N_{\bc_1},\emptyset)\oplus (N_{\bc_3},\{\removedp3\})$ nor 
%$(N_{\bc_2},\emptyset)\oplus (N_{\bc_3},\{\removedp3\})$ are
$N_{\bc_1}\oplus (N_{\bc_3},\{\removedp3\})$ nor 
$(N_{\bc_2},\{\removedp2\})\oplus (N_{\bc_3},\{\removedp3\})$ are
defined because $N_{\bc_3}$ shares the place $\removedp4$ with  $N_{\bc_1}$ and the place $\removedp6$ with  $N_{\bc_2}$.
Similarly, note that none of the considered occurrence nets can be composed sequentially, because their
interfaces do not match. For instance, the final place $\removedp5$ of  
%$(N_{\bc_1},\emptyset)\oplus(N_{\bc_2},\emptyset):\{\removedp1, \removedp2\}\to\{\removedp4,\removedp5, \removedp6\}$ 
$N_{\bc_1}\oplus (N_{\bc_2},\{\removedp2\}):\{\removedp1\}\to\{\removedp4,\removedp5, \removedp6\}$ 
does not appear as an initial place of $(N_{\bc_3},\{\removedp3\}): \{\removedp4,\removedp6\}\to\{\removedp7,\removedp8,\removedp9,\removedp{10}\}$.
We can fix this mismatch by considering the net $I_{\{\removedp5\}}:\{\removedp5\}\to\{\removedp5\}$ and noting that
%$(N_{\bc_3},\{\removedp3\}) \oplus (I_{\{\removedp5\}},\emptyset): \{\removedp4,\removedp6,\removedp5\}\to\{\removedp7,\removedp8,\removedp9,\removedp{10},\removedp5\}$ is well 
$(N_{\bc_3},\{\removedp3\}) \oplus I_{\{\removedp5\}}: \{\removedp4,\removedp6,\removedp5\}\to\{\removedp7,\removedp8,\removedp9,\removedp{10},\removedp5\}$ is well 
defined. Then, 
%$$((N_{\bc_1},\emptyset)\oplus(N_{\bc_2},\emptyset)); ((N_{\bc_3},\{\removedp3\}) \oplus (I_{\{\removedp5\}},\emptyset)):\{\removedp1, \removedp2\}\to\{\removedp5,\removedp7,\removedp8,\removedp9,\removedp{10}\}$$
$$N_{\bc_1}\oplus (N_{\bc_2},\{\removedp2\}) ; (N_{\bc_3},\{\removedp3\}) \oplus I_{\{\removedp5\}}:\{\removedp1\}\to\{\removedp5,\removedp7,\removedp8,\removedp9,\removedp{10}\}$$
stands for the net $N$ in Fig.~\ref{fig:N}.

\end{example}

A marked occurrence net is called \emph{trivial} if it has no transitions.

We say a marked occurrence net $\mathcal{M}$ is decomposable in parallel if there exists two non-trivial marked occurrence nets  $\mathcal{M}_1$ and $\mathcal{M}_2$ such that 
$\mathcal{M} = \mathcal{M}_1 \oplus \mathcal{M}_2$.
Similarly, we say that it is decomposable in series if there exists two non-trivial marked occurrence nets  $\mathcal{M}_1$ and $\mathcal{M}_2$ such that 
$\mathcal{M} = \mathcal{M}_1 ; \mathcal{M}_2$.

\begin{lemma}
Any s-cell $N_{\bc}$ cannot be decomposed in series and in parallel.
\end{lemma}
\begin{proof}
By contraposition, it is immediate to prove that the sequential/parallel composition of two non-trivial nets is not an s-cell.
\end{proof}

\begin{proposition}
\label{prop:decomposition}
Any marked occurrence net can be uniquely decomposed as the parallel and sequential composition of its s-cells (and identities), up to the axioms of commutative monoidal pre-categories.
\end{proposition}
\begin{proof}
For the existence, the partial order of s-cell (is unique and it) induces a decomposition of the net.
For instance this can be done by stratifying the s-cells in layers $L_1,...,L_n$ where each layer $L_j$ is the (largest) parallel composition of some identity $I_{s_j}$ with all s-cells whose predecessors are in layers $L_1,...,L_{j-1}$ and then taking their sequential composition $L_1;...;L_n$.

For uniqueness, suppose two different decompositions can be found, then they must have the same s-cells (because s-cells are not decomposable) ordered in the same way (because the ordering is induced by the places they share), hence they coincide.
\end{proof}

\begin{definition}[Canonical form]
Given a marked occurrence net  $\mathcal{M}$ we denote by $\mathsf{can}(\mathcal{M})$ its unique decomposition.
\end{definition}

\begin{example} 
\label{ex:canonical-of-N}
The canonical form of  $(N,\{\removedp2,\removedp3\})$ in Fig.~\ref{fig:N} is given by the decomposition below, already discussed in Example~\ref{ex:decomposition}:
%$$((N_{\bc_1},\emptyset)\oplus(N_{\bc_2},\emptyset)); ((N_{\bc_3},\{\removedp3\}) \oplus (I_{\{\removedp5\}},\emptyset)):\{\removedp1, \removedp2\}\to\{\removedp5,\removedp7,\removedp8,\removedp9,\removedp{10}\}$$
$$N_{\bc_1}\oplus (N_{\bc_2},\{\removedp2\}); (N_{\bc_3},\{\removedp3\}) \oplus I_{\{\removedp5\}}:\{\removedp1\}\to\{\removedp5,\removedp7,\removedp8,\removedp9,\removedp{10}\}$$

\end{example}

\subsection{Place Removal}

Given a possibly marked s-cell $N_\bc:i\to o$ (with $i\neq \emptyset$), we are interested in studying what happens under the hypothesis that some tokens arrive in a subset of places $m\subseteq i$ while the places in $s=i\setminus m$ are guaranteed to stay empty (i.e., they are \emph{dead}).
In fact it can happen that the removal of the places in $s$ and of the transitions and places that causally depend on them\footnote{In such cases, all the transitions that depend on some place in $s$ cannot be fired and  the places in their post-set are also dead.} will allow to further decompose the s-cell.

%Given a possibly marked s-cell $N_\bc:i\to o$ with $i\neq \emptyset$ we are interested in studying what happens under the hypothesis that some tokens cannot arrive in a set of places $s\subseteq i$, i.e. under the hypothesis that $s$ are dead places. In such cases, all the transitions that depends on some place in $s$ cannot be fired and  the places in their post-set are also dead.

We let $N_\bc \ominus s$ be the net obtained by removing all dead nodes as explained above.
Additionally, isolated places are also removed.
The cancellation of some transitions can break the equivalence class induced by $\sqsubseteq$, which explains why $N_\bc \ominus s$ is not necessarily an s-cell. Also note that some of the final places of $N_\bc$ can become dead and canceled.
The final dead places can be computed by taking $\maxp{N_\bc} \setminus \maxp{(N_\bc \ominus s)}$.
Thus in general we have $N_\bc \ominus s: i' \to o'$ for some $i' \subseteq i\setminus s$ and $o' \subseteq o$.
%
%Similarly, we are interested in studying what happens under the hypothesis that some tokens arrive in a set of places $m\subseteq i$ while the places in $s=i\setminus m$ are dead. 
We write $N_\bc @ m$ for the marked net $(N_\bc\ominus s,\minp{(N_\bc\ominus s)}):\emptyset \to o'$,
where $N_\bc:i\to o$ and $s=i\setminus m$,
i.e., for the net $N_\bc \ominus s$ whose initial places are all marked.

To some extent the behaviour of an s-cell is  determined by considering its behaviour under all possible initial markings. Consequently we can further explore the behaviour of $N_\bc:i\to o$ by considering  $N_\bc @ m$ for all $m\subseteq i$. 

%This is written as 
%\todo{forse questo non serve, non aiuta a capire?}
%$$N_\bc \equiv \sum_{m\subseteq i} \mathsf{can}(N_\bc @ m)$$

\begin{example} Consider the s-cell $(N_{\bc_3},\{\removedp3\}) : \{\removedp4, \removedp6\}\to\{\removedp7,\removedp8,\removedp9,\removedp{10}\}$ in Fig.~\ref{fig:N-bc-3}. 
The behaviour of $(N_{\bc_3},\{\removedp3\})$ can be explained by considering all the possible ways in which 
its initial places $\removedp4$ and $\removedp6$ can be marked: none of them is marked (i.e., $N_{\bc_3}@\{\removedp3\}$),
just one of them is marked (i.e., either $N_{\bc_3}@\{\removedp3,\removedp4\}$ or $N_{\bc_3}@\{\removedp3,\removedp6\}$), 
or both of them are marked (i.e.,  
$N_{\bc_3}@\{\removedp3,\removedp4,\removedp6\}$). 
Net $N_{\bc_3}@\{\removedp3\}$ depicted in Fig.~\ref{fig:N-bc-3-at-p3} is obtained by removing from $N_{\bc_3}$  the initial places 
$\removedp4$ and $\removedp6$, and all the elements that causally depends on them, i.e., the transitions 
$f$, $g$ and $h$ and the places $\removedp7$, $\removedp8$, $\removedp9$ and $\removedp{10}$.
The remaining nets are in Fig.~\ref{fig:N-bc-3-at-p3p4}-\ref{fig:N-bc-3-at-p3p4p6}.
It is worth noticing that in $N_{\bc_3}@\{\removedp3,\removedp4\}$ the place $\removedp4$ is also removed from $N_{\bc_3}\ominus \{\removedp6\}$ because, after removing the place $\removedp6$ and thus the transition $f$, the place $\removedp4$ remains isolated.
\end{example}

% !TEX root =  main.tex

\section{Compiling nets}\label{sec:PNtoBN}
In this section we associate each finite occurrence net with an arrow in the Kleisli category $\Kl\D$ of discrete probability distributions. This is 
achieved in two steps. We first introduce a language for representing occurrence nets and show how the 
s-cell decomposition can be used to  associate each occurrence net with a particular term. Then, we 
map terms into arrows in  $\Kl\D$.
%\todo[inline,caption={}]{
%The closure construction on a net returns a partial ordering of branching cells. If we assume to have an arrow of the distribution category for each branching cell, the string diagram of the monoidal structure determines a unique resulting arrow. To build the arrow corresponding to a branching cell we must possibly evaluate the cell for all initial markings.
% Given an initial marking, the cancellation of unreachable transitions will transform the cell into a net, for which we apply recursively the previous step. 
% If the net contains a unique branching cell the corresponding arrow is trivial. Not necessarily the construction should consider every initial marking, certain rows could remain unspecified.
%}

\subsection{Language of nets}

The decomposition of a net in branching cells can be described by terms generated by the following grammar, where $m,s$ are sets of places and $\Theta$ is a set of transactions:

%\[ T :: =   \bct\ |\ T\ten T\ |\ T;T\ |\ \sum_{i\in 1..n}\ T_i\]  
%In $\bct$, $i$ denotes the  minimal places (marked), $o$ the final places, 
%and $t\subseteq 2^o$ stands for the  transactions of the branching cell. 

% NEW
\[ T :: =   I_s\ |\ \bot_s\ |\ T\oplus T\ |\ T;T\ |\ \bct\ |\ \sum_{m\subseteq s}\ \cop[m][T]\]  

Here the idea is that $\bct$ denotes a basic building block consisting of the set of transactions of an s-cell whose initial places are all marked. 
The case of an s-cell $\bc$ with a set of unmarked initial places $s$ is represented as the formal sum $ \sum_{m\subseteq s}\ \cop[m][T]$, where all the possibile ($2^{|s|}$) initial markings $m$ are considered, each paired with the encoding of $N_\bc @ m$.
The term $I_s$ denotes the identity net, consisting just of a set of unmarked places with no transitions (i.e., all places are initial and final).
The term $\bot_s$ denote a net with no initial places and no transitions, whose only final places are $s$ (i.e., the places $s$ are dead).
The terms $T\oplus T$ and $T;T$ denote respectively the composition in parallel and in series.
%
%The term $\bct$ stands for an s-cell whose initial places are all marked.
%In $\bct$, $m$ denotes the initial marked places, $o$ the final places, 
%and $\Theta$ stands for set of transactions of the branching cell, where for each $\theta:m'\to o'\in \Theta$ we require that $m'\subseteq m$ and $o' \subseteq o$.
%
%As a sanity check, as transactions are maximal processes, we also require that for any $\theta_1:m_1\to o_1,\theta_2:m_2\to o_2\in \Theta$ we have $m_1 \cup m_2 \not\subseteq m$.
% 
%Finally, the term $\sum_{m\subseteq i}\ \cop$ accounts for an s-cell whose unmarked initial places are $i$, by expressing it as the combination of $2^{|i|}$ terms, one for each possibile marking of the places in $i$.

The terms of the algebra are taken up to the axioms of commutative monoidal (pre-)categories, where additionally we have $\bot_\emptyset = I_\emptyset$ and $\bot_{s_1\cup s_2} = \bot_{s_1} \oplus \bot_{s_2}$.

\subsubsection{Typing} 

%\[
%\begin{array}{l}
%\reductionRule{%i\cap o = \emptyset  \quad 
%     t\subseteq 2^o}
%{ \typeJudgment{}{\bct}{\sort {\emptyset}{o}[i]}} 
%\\[20pt]
%\reductionRule{ \typeJudgment{}{T}{\sort {i}{o}[m]} \quad \typeJudgment{}{T'}{\sort{i'}{o'}[m'] } 
%\quad i \cap i' = \emptyset      \quad o \cap o' = \emptyset  \quad m \cap m' = \emptyset }
%{ \typeJudgment{}{T\ten T'}{\sort{i\cup i'}{o\cup o'}[m\cup m']}} 
%\\[20pt]
%\reductionRule{ \typeJudgment{}{T}{\sort { i}{o}[m]} \quad \typeJudgment{}{T'}{\sort {o}{o'}[m']}% \quad  i\cap o' =\emptyset
%}
%{ \typeJudgment{}{T\comp T'}{\sort {i} {o'}[m\cup m']}}
%\\[20pt]
%\reductionRule{ n = 2^{|i|}  \ \wedge\  \forall j.(\typeJudgment{}{T_j}{\sort{\emptyset}{o}[m\cup i_j] } \ \wedge\  i_j \subseteq i   
%\ \wedge\ \forall k\neq j (\typeJudgment{}{T_k}{\sort {\emptyset}{o}[i_k\cup m]} \wedge i_j \neq i_k))}
%{ \typeJudgment{}{\sum_{j\in 1..n}\ T_i}{\sort { i}{ o}[m]}}
%\end{array}
%\]

% NEW
Not all terms are valid though. We introduce a type system to discard ill-formed terms.
Our types are triples of the form $(i,s,o)$ where $i$  is the set of initial unmarked places, $s$ is the set of all places and transitions appearing in a term and $o$ is the set of final places. %We leave it implicit that $m$ and $i\cup o$ are disjoint and that $i\cup m\cup o \subseteq s$.

We write $T:\termtype{i}{m}{s}{o}$ for $T:(i,s,o)$. 
The typing rules are in Fig.~\ref{fig:types}. The rules for $I_s$ and $\bot_s$ are self-explanatory.
The rule for $\oplus$ states that a term is well-typed when its subterms are well-typed and do not 
share place nor transitions (i.e., $s\cap s' = \emptyset$). The case of 
sequential composition $T;T'$ additionally requires that the set of
final places of $T$ coincides with the set of the initial unmarked 
places of $T'$. The rule for $\sum_{m\subseteq i}\cop$ requires all subterms $T_m$ to have 
the same sets of initial and final places (respectively, $\emptyset$ and $o$), which captures 
the idea that a sum represents the execution of a s-cell under all possible markings. 
The rule for $\bct$ follows immediately.

\begin{figure}
\[
\begin{array}{c}
\reductionRule{\phantom{T':\termtype{i'}{m'}{s'}{o'}}}
{I_s:\termtype{s}{\emptyset}{s}{s}}
\qquad
\reductionRule{\phantom{T':\termtype{i'}{m'}{s'}{o'}}}
{\bot_s:\termtype{\emptyset}{\emptyset}{s}{s}}
\qquad
\reductionRule{T:\termtype{i}{m}{s}{o}\quad T':\termtype{i'}{m'}{s'}{o'}\quad s\cap s'=\emptyset} 
{ T\oplus T':\termtype{i\cup i'}{m\cup m'}{s\cup s'}{o\cup o'}} 
\\[30pt]
\reductionRule{T:\termtype{i}{m}{s}{m}\quad T':\termtype{m}{m'}{s'}{o}\quad s\cap s'=m} 
{ T; T':\termtype{i}{m\cup m'}{s\cup s'}{o}} 
\qquad
\reductionRule{\forall m\subseteq i.~ T_{m}:\termtype{\emptyset}{m\cup o}{s_m}{o} \quad s=\bigcup_{m\subseteq i}s_{m}}
{ \sum_{m\subseteq i}\ \cop:\termtype{i}{m}{s}{o}}
\\[30pt]
%\reductionRule{m\cap o = \emptyset\quad \forall \theta:m'\to o'\in \Theta. m'\subseteq m\wedge o'\subseteq o}
\reductionRule{o = \bigcup_{\theta\in \Theta} \maxp{\theta} \qquad s = \bigcup_{\theta\in \Theta} {\bf n}({\theta}) }
%\quad \forall \theta:m'\to o'\in \Theta. m'\subseteq m\wedge o'\subseteq o}
{ \bct:\termtype{\emptyset}{m}{s}{o}}
\end{array}
\]
\caption{Type system}\label{fig:types}
\end{figure}

%\begin{lemma} If $\typeJudgment{}{T}{\sort {i_1}{o_1}[m_1]}$ and $\typeJudgment{}{T}{\sort {i_2}{o_2}[m_2]}$ then
%$i_1=i_2$, $o_1 = o_2$ and $m_1=m_2$.
%\end{lemma}
%
%Moreover we consider terms up-to equivalence ...
%
%\begin{lemma} If $T_1 \equiv T_2$ and $\typeJudgment{}{T_1}{\sort {i}{o}[m]}$ then $\typeJudgment{}{T_2}{\sort {i}{o}[m]}$.
%\end{lemma}
%
%We write $\idN[]^a = \bct[\emptyset][\{a\}][\emptyset] +\bct[\{a\}][\{a\}][\{a\}]$. Also
%$\idN[{a_1\ldots a_n}] =  \idN[{a_1}]\ten\idN[^{a_2\ldots a_n}]$.

\begin{lemma} If $T:\termtype{i}{m}{s}{o}$ then
$i\cup o\subseteq s$.
\end{lemma}
\begin{proof}
The proof is by rule induction.
\end{proof}

%NEW
Typing is unique, as stated by the following result. 
\begin{lemma} If $T:\termtype{i}{m}{s}{o}$ and $T:\termtype{i'}{m'}{s'}{o'}$ then
$i=i'$, $o=o'$, $s=s'$.% and $o=o'$.
\end{lemma}
\begin{proof}
The proof is by rule induction.
\end{proof}

Hereafter we assume terms to be well-typed.

\subsection{From Nets to Terms}

In this section we introduce a mapping from occurrence nets to terms. 

\begin{definition}
   Let $\mathcal{M}$ be a marked occurrence net.
   % and let $\mathcal{C}$ be its canonical normal form. 
   The corresponding term $\termsof{\mathcal{M}}$ is given by the homomorphic extension (w.r.t. identitites, parallel and sequential composition)\footnote{This just means that 
$\termsof{I_s} = I_s$, 
$\termsof{\mathcal{M}_1\oplus \mathcal{M}_2} =  
\termsof{\mathcal{M}_1} \oplus \termsof{\mathcal{M}_2}$ and
$\termsof{\mathcal{M}_1; \mathcal{M}_2} =  
\termsof{\mathcal{M}_1} ; \termsof{\mathcal{M}_2}$.} of the encoding defined below over s-cells.  

       \begin{subnumcases}{\termsof{N_{\bc},i}}
           \bct[m][o][\Theta(N_{\bc})] & if $\minp{N_{\bc}} = i$\hfill{\ } \label{case:toT-bc-base}\\
           \displaystyle\sum_{m\subseteq \minp{(N_{\bc},i)}} \cop[m][(\bot_{d_{m}} \oplus T_{m} )] & otherwise \label{case:toT-bc-decomp}
      \end{subnumcases}                 

      \hspace{2cm} 
      {where:}
            $\left\{
               \begin{array}{l@{\ }l}
               N_{m} &= N_{\bc}@i\cup m\\
               T_{m} &= \  \termsof{\mathsf{can}(N_m)}\\
               d_{m} &= \  \maxp{N_{\bc}} \setminus \maxp{N_m}\\
               \end{array}\right.$

\end{definition}

The encoding of a marked s-cell $\bc$ considers two cases: (i) all initial places of the s-cell are marked (Eq.~\ref{case:toT-bc-base}); and 
(ii) some initial tokens are unmarked. In the first case, a completely marked s-cell is mapped to the term $\bct[m][o][\Theta(N_{\bc})]$ that
describes all the possible executions of $N_{\bc}$, i.e., its transactions. Differently, when some initial places are unmarked, 
the corresponding term is obtained by composing the behaviour of the s-cell under each possible marking ${m\subseteq \minp{(N_{\bc},i)}}$.
The term $\cop[m][(\bot_{d_{m}} \oplus T_{m})]$ describes the behaviour of $\bc$ when all places 
in $i\cup m$ are marked and the remaining initial places are dead. For this reason, $\bot_{d_{m}}$  and $T_{m}$
are defined in terms of the net $N_{m} = N_{\bc}@i\cup m$. The term $\bot_{d_{m}}$ stands for the 
final places that are dead when the initial marking is $i\cup m$. 
The term $T_m$ encodes the net $N_{\bc}@i\cup m$: we just remark
here, as already mentioned, that we need to compute the canonical form of $N_m$, because removing elements from $\bc$ may 
originate a complex net an not an s-cell (as for $N_{\bc_3}@\{\removedp3,\removedp6\}$ in Fig.~\ref{fig:N-bc-3-at-p3p6}).

\begin{lemma}   For any finite occurrence net $N$  and marking $m\subseteq \minp N$,
$\toT$ is defined, unique (up-to the structure of commutative monoidal pre-categories) and well-typed. 
\end{lemma}

\begin{example}
\label{ex:to-terms-N}
Consider the marked occurrence net $(N,\{\removedp2, \removedp3\})$ in Fig.~\ref{fig:N}, whose canonical form is in Example~\ref{ex:canonical-of-N}
$$(N,\{\removedp2, \removedp3\})= N_{\bc_1}\oplus (N_{\bc_2},\{\removedp2\}); (N_{\bc_3},\{\removedp3\}) \oplus I_{\{\removedp5\}}$$
Then, the corresponding term is obtained by
\begin{equation}
\label{eq:terms-N}
\termsof{N,\{\removedp2, \removedp3\}}= \termsof{N_{\bc_1}}\oplus \termsof{N_{\bc_2},\{\removedp2\}}; \termsof{N_{\bc_3},\{\removedp3\}} \oplus \termsof{I_{\{\removedp5\}}}
\end{equation}

The term
$\termsof{N_{\bc_1}}$ is obtained by applying Eq.~\eqref{case:toT-bc-decomp}  
because  $i=\emptyset$ and $\minp{N_{\bc_1}} = \{\removedp1\} \neq \emptyset$ (see  $N_{\bc_1}$ in Fig.~\ref{fig:N-bc-1}).
Then, 
\begin{equation}
\termsof{N_{\bc_1}} = 
    \cop[\emptyset][(\bot_{d_\emptyset} \oplus T_{\emptyset})]
  +\cop[\{\removedp1\}][(\bot_{d_{\{\removedp1\}}} \oplus T_{\{\removedp1\}})]
\end{equation}
  
 Note that $N_\emptyset = N_{\bc_1}@\emptyset$  is obtained
 from $N_{\bc_1}$  by removing all elements  that depends on the unique unmarked initial place $\removedp1$. Hence, 
   $N_\emptyset = N_{\bc_1}@\emptyset = \mathbf{0} = I_{\emptyset}$. Consequently, $T_{\emptyset} = \termsof{N_m} = I_{\emptyset}$.
  Moreover ${d_\emptyset} = \{\removedp4, \removedp5\}$.  
  
  For the marking $\{\removedp1\}$, we have $N_{\{\removedp1\}} = N_{\bc_1}@\{\removedp1\} = (N_{\bc_1},\{\removedp1\})$. Since  $N_{\bc_1}$ is an s-cell, 
    ${\mathsf{can}(N_{\bc_1}@\{\removedp1\})} = {(N_{\bc_1}, \{\removedp1\})}$. Therefore, 
    $T_{\{\removedp1\}} =  \termsof{N_{\bc_1}, \{\removedp1\}}$, which is obtained by using Eq.~\eqref{case:toT-bc-base}. 
   The net $N_{\bc_1}$ has two transactions, one for each transition, i.e., $\Theta(N_{\bc_1}) = \{\{a\},\{b\}\}$.
   Then, $T_{\{\removedp1\}} = \bct[\{\removedp1\}][\{\removedp4,\removedp5\}][\{\{a\},\{b\}\}]$. Moreover,    
   ${d_{\{\removedp1\}}} = \emptyset$ because $\maxp{(N_{\{\removedp1\}})} = \maxp{(N_{\bc_1},\{\removedp1\})} = \maxp{N_{\bc_1}}$.
   Consequently,  
\begin{equation}
\label{eq:terms-bc1}
  \begin{array}{r@{\ }l}
  \termsof{N_{\bc_1}} = &
    \cop[\emptyset][(\bot_{\{\removedp4, \removedp5\}} \oplus I_{\emptyset})]
  +\cop[\{\removedp1\}][{(\bot_{\emptyset} \oplus \bct[\{\removedp1\}][\{\removedp4,\removedp5\}][\{\{a\},\{b\}\}])}]
  \\ = &
   \cop[\emptyset][\bot_{\{\removedp4, \removedp5\}}]
  +\cop[\{\removedp1\}][{\bct[\{\removedp1\}][\{\removedp4,\removedp5\}][\{\{a\},\{b\}\}]}]
  \end{array}
\end{equation}
  Intuitively, the term $\cop[\emptyset][\bot_{\{\removedp4, \removedp5\}}]$ 
   states that 
  the s-cell $\bc_1$ does not generate any token in its final places when the
  initial place $\removedp1$ remains unmarked.
  Differently,  $\cop[\{\removedp1\}][{\bct[\{\removedp1\}][\{\removedp4,\removedp5\}][\{\{a\},\{b\}\}]}]$ 
  describes the behaviour of $\bc_1$ when its initial place is marked. In this case, 
  the behaviour corresponds to the non-deterministic choice of the transactions $\{a\}$ and $\{b\}$.
  
  The encoding of $(N_{\bc_2},\{\removedp2\})$ is obtained by using Eq.~\eqref{case:toT-bc-base},
 \begin{equation} 
\label{eq:terms-bc2}
  \termsof{N_{\bc_2},\{\removedp2\}} = {\bct[\{\removedp2\}][\{\removedp6\}][\{\{c\},\{d\}\}]}
\end{equation}  

For ${(N_{\bc_3}, \{\removedp3\})}$, we obtain the following term by analogous calculations
\begin{equation}
\label{eq:terms-bc3}
\begin{array}{l@{\ }c@{\ }lll}
\termsof{N_{\bc_3},\{\removedp3\}} = &&
    \cop[\emptyset][(\bot_{\{\removedp8,\removedp9,\removedp{10}\}} \oplus {\bct[\{\removedp3\}][\{\removedp7\}][\{\{e\}\}]})]
    \\
  &+& \cop[\{\removedp4\}][{(\bot_{\{\removedp8,\removedp9,\removedp{10}\}} \oplus {\bct[\{\removedp3\}][\{\removedp7\}][\{\{e\}\}]})}]
    \\
  &+&
  \cop[\{\removedp6\}][{(\bot_{\{\removedp8\}} \oplus  \bct[{}][\{\removedp7\}][{\{\{e\}\}}] \oplus \bct[{}][\{\removedp9,\removedp{10}\}][{\{\{g\},\{h\}\}}])}]
    \\
  &+&\cop[\{\removedp4,\removedp6\}][{\bct[{}][\{\removedp{8},\removedp{9},\removedp9,\removedp{10}\}][{\{\{f\},\{e,g\},\{e,h\}\}}]}]\\
  \end{array}
\end{equation}
which describes the behaviour of $\bc_3$ for every possible initial
marking of its initial places (i.e., $\emptyset$, $\{\removedp4\}$, $\{\removedp6\}$,
and $\{\removedp4,\removedp6\}$). The most interesting case is the subterm 
$\cop[\{\removedp6\}][{(\bot_{\{\removedp8\}} \oplus  \bct[{}][\{\removedp7\}][{\{\{e\}\}}] \oplus \bct[{}][\{\removedp9,\removedp{10}\}][{\{\{g\},\{h\}\}}])}]$
obtained from $\cop[\{\removedp6\}][{(\bot_{d_{\{\removedp6\}}} \oplus T_{\{\removedp6\}})}]$. Consider
the net $N_{\{\removedp6\}} =  (N_{\bc_3}@\{\removedp3,\removedp6\})$ in Fig.~\ref{fig:N-bc-3-at-p3p6}, which
contains two s-cells. Consequently, its  
canonical form is given by the parallel composition of two s-cells, which are respectively encoded as 
 $\bct[{}][\{\removedp7\}][{\{\{e\}\}}]$ and $\bct[{}][\{\removedp9,\removedp{10}\}][{\{\{g\},\{h\}\}}]$. 
 
 Finally, 
\begin{equation}
\label{eq:terms-I5}
\termsof{I_{\{\removedp5\}}} = I_{\{\removedp5\}}
\end{equation}

%
%$$
%\begin{array}{l@{\ }l}
%T_{\{\removedp4\}} =  \bct[{}][\{\removedp7\}][{\{\{e\}\}}] \oplus \bct[{}][\{\removedp9,\removedp{10}\}][{\{\{g\},\{h\}\}}]
%\end{array}
%$$
%
%\item ${d_{\{\removedp6\}}} =   \bot_{\{\removedp8\}}$.
 
\end{example}

To show that the term $\termsof{N,m}$ is a good representative of the probabilistic semantics of $N$, we prove that it characterises the configurations allowed by the semantics of Abbes and Benveniste. The interested reader can find all technical details in the Appendix.

\subsection{From Terms to $\Kl\D$}

Given a set $X$, a discrete probability distribution with finite support over $X$ is a function $\omega: X \to [0,1]$ such that $\sum_{x\in X}^n \omega(x) = 1$ and $\mathsf{supp}(\omega) = \{x\in X \mid \omega(x)>0\}$ is a finite set. The function $\omega$ can be sometimes written as the formal convex combination\footnote{The `ket' notation $\ket{r}{x}$ has no particular meaning: it is just syntactic sugar.}
$$
\omega = \ket{r_1}{x_1} + ... + \ket{r_n}{x_n}
$$
where $\mathsf{supp}(\omega) = \{x_1,...,x_n\}$ and $r_j = \omega(x_j)$ for $j\in[1,n]$.
We let $\D(X)$ be the set of discrete probability distributions $\omega$ over $X$ and write $\D$ for the discrete probability monad over the category $\mathbf{Set}$ of sets (as objects) and functions (as arrows). The category $\Kl\D$ is the Kleisli category of the monad $\D$: its objects are sets, its arrows  $f:X \to Y$ are functions $f:X\to \D(Y)$.
It has been shown in~\cite{DBLP:journals/entcs/JacobsZ16} that $\Kl\D$ forms a symmetric monoidal category and that Bayesian networks can be seen as special kinds of arrows in $\Kl\D$ that can be represented as string diagrams using wire-and-box notation.
According to this view, a diagram from $n$ to $k$ represents an arrow from $2^n$ to $2^k$ in $\Kl\D$.

We next show how to interpret Petri nets as Bayesian networks by exploiting $\Kl\D$.
To this aim we need to map the arrows of a commutative pre-monoidal pre-category to those of a symmetric monoidal category: in the first case the objects are sets of places, while in the latter they are natural numbers representing a totally ordered set of ports.
Therefore the mapping is defined parametrically on some arbitrarily chosen total orders of initial and final places. 

Given a set of places $s$, we let $\pi_s$ denote a bijective function $\pi_s : s \to |s|$ that assigns a position to each element of $s$. We write $\pi$ when the set $s$ is implicit. Overloading the notation, we let $\pi$ also denote  the string such that the place $p\in s$ appears in position $\pi(p)$. Note that $\pi$ is without repetitions: each $p\in s$ appears exactly once in $\pi$. We let $\epsilon$ denote the empty string (over the empty set of places). For $p\in s$ and $m\subseteq s$, we also write $p\in \pi$ and $m\subseteq \pi$ when $\pi$ is a linearization of $s$.

%We write $\cseq s$ for a total ordering of the elements of a set $s$, i.e., a sequence containing all and only the elements in $s$. 
%Given a sequence $\cseq s$, we assume the order of $\cseq s$ to be extended to a total order over  $2^s$ (e.g., lexicographically). 

Given $\pi$ and $\pi'$ two such strings over $s$, we let $\chi^{\pi}_{\pi'}: |s| \to |s|$ denote the unique permutation that swaps $\pi$ into $\pi'$, i.e. such that for any $p\in s$ we have $\chi^{\pi}_{\pi'}(\pi(p)) = \pi'(p)$. By coherence of symmetries we have, e.g., $\chi^{\pi}_{\pi'};\chi^{\pi'}_{\pi''} = \chi^{\pi}_{\pi''}$.

Given two strings $\pi$ over $s$ and $\pi'$ over $s'$ with $s\cap s'=\emptyset$ we use juxtaposition to denote the string $\pi\pi'$ over $s\cup s'$ such that 
$(\pi\pi')(p) = \pi(p)$ if $p\in s$ and 
$(\pi\pi')(p) = |s|+\pi'(p)$ if $p\in s'$.

As a matter of notation, we assume that a string $\pi$ over $s$ implicitly defines an ordering over $\Bin^s$, e.g., a subset of $s$ can be seen as a binary string of length $|s|$, which are then ordered lexicographically. Correspondingly, the permutation $\chi^{\pi}_{\pi'}: |s| \to |s|$ induces an isomorphism on $\Bin^s$, that we denote with the same name $\chi^{\pi}_{\pi'}$.

In the following we assume a function $\delta$ is given that associates every constant $\bct$ with 
a finite discrete probability distribution over the elements in $\Theta$. 
To ease readability, we write $\delta_{\bct}$ for the probability distribution $\delta(\bct)$ over $\Theta$.

\begin{definition}
Let  $T:\termtype{i}{m}{s}{o}$ be a well-typed term, $\pi$ a string over $i$, $\rho$ a string over $o$. Then, 
 $\kleisli{T}{\delta}{\pi}{\rho}$ stands for an arrow $2^{|i|}\rightarrow 2^{|o|}$ in $\Kl\D$ (i.e., a diagram from $|i|$ to $|o|$) defined by structural induction as follows:

%\begin{align}
%\label{eq:bc}
% \enc[\bct] & = \delta(\bct) \\
%\label{eq:ten}
%\enc[T_1\ten T_2] & = \chi^{\cseq {i}}_{\cseq{i}_1\cseq{i}_2}; (\enc[T_1][\cseq {i_1}][\cseq {o_1}]\ten\enc[T_2][\cseq {i_2}][\cseq {o_2}]);\chi^{\cseq{o}_1\cseq{o}_2}_{\cseq o}\\
%\label{eq:seq}
%\enc[T_1;T_2] & =  \enc[T_1][\cseq i][\cseq m];\enc[T_2][\cseq m]\\
%\begin{split}
%\hspace{-5cm}\enc[\sum_{i\in 1..n}T_i] & = [\enc[T_{j_1}],\ldots,\enc[T_{j_n}]] \label{eq:sum}\\
%& \mbox{where}\  {j_1\ldots j_n} \mbox{is a permutation of}\ 1..n, \mbox{and}
%\\ & \forall k.\typeJudgment{}{T_{j_k}}{\sort {i_{j_k}} {o_{j_k}}} \mbox{ and}\ k<l \implies {\cseq i_{j_k}}< {\cseq i_{j_l}}.\\
%\end{split}
%\end{align}
%
%Eq.~\ref{eq:seq} is well-defined because well-typed terms ensure that the interfaces match and the 
%encoding is defined for any sequence over the places in the interface.

% NEW
\begin{eqnarray}
\kleisli{I_s}{\delta}{\pi}{\rho} 
& = & 
\chi^{\pi}_{\rho}
\label{eq:id} \\
\kleisli{\bot_s}{\delta}{\epsilon}{\rho} 
& = & 
\delta^{|s|}_0 
\label{eq:bot} \\
\kleisli{T_1 \oplus T_2}{\delta}{\pi}{\rho} 
& = & 
\chi^{\pi}_{\pi_1\pi_2};(\kleisli{T_1}{\delta}{\pi_1}{\rho_1} \otimes \kleisli{T_2}{\delta}{\pi_2}{\rho_2});\chi^{\rho_1\rho_2}_{\rho} 
\label{eq:par} \\
\kleisli{T_1;T_2}{\delta}{\pi}{\rho} 
& = & 
\kleisli{T_1}{\delta}{\pi}{\gamma};\kleisli{T_2}{\delta}{\gamma}{\rho}
\label{eq:seq} \\
\kleisli{\bct}{\delta}{\epsilon}{\rho} 
& = & 
%\lambda m.~\sum_{\theta:\emptyset\to m\in \Theta}\delta(\bct)(\theta) 
\lambda m.~\sum_{\theta:\emptyset\to m\in \Theta}\delta_{\bct}(\theta)
\label{eq:bct} \\
\kleisli{\sum_{m\subseteq i}\cop}{\delta}{\pi}{\rho} 
& = & 
[\kleisli{T_{\pi^{-1}(1)}}{\delta}{\epsilon}{\rho},..., \kleisli{T_{\pi^{-1}(2^{|i|})}}{\delta}{\epsilon}{\rho}] \label{eq:sum}
\end{eqnarray}
where in Eq.~\eqref{eq:bot} the probability distribution $\delta^{|s|}_0$ assigns probabilty $1$ to the case $\emptyset$ and $0$ to all the remaining $2^{|s|}-1$ cases and in Eq.~\eqref{eq:sum} the arrows is obtained as the copairing of each $T_m$ for all $m\subseteq i$.\footnote{It is important to mention that in Eq.~\eqref{eq:sum} the order of the arrows in the copairing is the one induced by $\pi$: remember that $\pi$ induces an order on $\Bin^i$, then $\pi^{-1}(k)$ denotes the $k$-th subset $m\subseteq i$ according to the order in $\pi$.}
\end{definition}

The cases in Eqs.~\eqref{eq:id} and~\eqref{eq:bot} are straightforward.
The cases in Eqs.~\eqref{eq:par} and~\eqref{eq:seq} just exploit the monoidal category structure.
It is worth noting that while the operation $\oplus$ is commutative, this is not the case for the monoidal operation of the Kleisli category, hence denoted with a different symbol $\otimes$.
The case in Eq.~\eqref{eq:bct} is the most interesting: $\kleisli{\bct}{\delta}{\epsilon}{\rho}$ must assign a probability distribution to the elements in the powerset of the places in $\rho$; given $m\subseteq \rho$ its probability is computed by taking the sum of the probabilities assigned by $\delta$ to all processes $\theta$ whose final places are exactly $m$. This is correct as any two such processes are mutually exclusive alternatives.
Finally, the case in Eq.~\eqref{eq:sum} is the most complex, as it exploits the hierarchical decomposition of s-cells. Here we take each $T_m$ and compute $2^{|i|}$ arrows $\kleisli{T_m}{\delta}{\epsilon}{\rho}: 2^0 \to 2^{|\rho|}$. Then, via co-pairing we get an arrow from $2^{|i|}$  to $2^{|\rho|}$. The order of the arrows in the co-pair expression is important to associate them to the right element $m\subseteq i$ (according to the order induced by $\pi$).

\begin{proposition}\label{prop:order}
$\kleisli{T}{\delta}{\pi}{\rho} = \chi^{\pi}_{\pi'};\kleisli{T}{\delta}{\pi'}{\rho'};\chi^{\rho'}_{\rho}$.
\end{proposition}
\begin{proof}
The proof is by structural induction on $T$.

For the case $T=\bot_s$, we have 
$\chi^{\epsilon}_{\epsilon};\kleisli{\bot_s}{\delta}{\epsilon}{\rho'};\chi^{\rho'}_{\rho} =
\kleisli{\bot_s}{\delta}{\epsilon}{\rho'};\chi^{\rho'}_{\rho} =
\delta^{|s|}_0;\chi^{\rho'}_{\rho} =
\delta^{|s|}_0$.

For the case $T=I_s$, we have 
$
\chi^{\pi}_{\pi'};\kleisli{I_s}{\delta}{\pi'}{\rho'};\chi^{\rho'}_{\rho} =
\chi^{\pi}_{\pi'};\chi^{\pi'}_{\rho'};\chi^{\rho'}_{\rho} =
\chi^{\pi}_{\rho}$
by coherence of symmetries.

For the case $T=T_1 \oplus T_2$, we have
\begin{eqnarray*}
 \chi^{\pi}_{\pi'};\kleisli{T_1 \oplus T_2}{\delta}{\pi'}{\rho'};\chi^{\rho'}_{\rho}
& = &
\chi^{\pi}_{\pi'};\chi^{\pi'}_{\pi_1\pi_2};(\kleisli{T_1}{\delta}{\pi_1}{\rho_1} \otimes \kleisli{T_2}{\delta}{\pi_2}{\rho_2});\chi^{\rho_1\rho_2}_{\rho'};\chi^{\rho'}_{\rho}\\
& = &
\chi^{\pi}_{\pi_1\pi_2};(\kleisli{T_1}{\delta}{\pi_1}{\rho_1} \otimes \kleisli{T_2}{\delta}{\pi_2}{\rho_2});\chi^{\rho_1\rho_2}_{\rho}\\
& = &
\kleisli{T_1 \oplus T_2}{\delta}{\pi}{\rho}
\end{eqnarray*}
by coherence of symmetries.

For the case $T=T_1;T_2$, let us assume that
$\kleisli{T_1}{\delta}{\pi_1}{\rho_1} = \chi^{\pi_1}_{\pi'_1};\kleisli{T_1}{\delta}{\pi'_1}{\rho'_1};\chi^{\rho'_1}_{\rho_1}$ and
$\kleisli{T_2}{\delta}{\pi_2}{\rho_2} = \chi^{\pi_2}_{\pi'_2};\kleisli{T_2}{\delta}{\pi'_2}{\rho'_2};\chi^{\rho'_2}_{\rho_2}$, so that, as a particular case we have
$\kleisli{T_1}{\delta}{\pi}{\gamma} = \chi^{\pi}_{\pi'};\kleisli{T_1}{\delta}{\pi'}{\gamma}$
and
$\kleisli{T_2}{\delta}{\gamma}{\rho} = \kleisli{T_2}{\delta}{\gamma}{\rho'};\chi^{\rho'}_{\rho}$
(because $\chi^{\gamma}_{\gamma} = I_{|\gamma|}$).
Then we have
\begin{eqnarray*}
\chi^{\pi}_{\pi'};\kleisli{T_1 \oplus T_2}{\delta}{\pi'}{\rho'};\chi^{\rho'}_{\rho}
& = &
\chi^{\pi}_{\pi'}; 
\kleisli{T_1}{\delta}{\pi'}{\gamma};\kleisli{T_2}{\delta}{\gamma}{\rho'};
\chi^{\rho'}_{\rho}\\
& = & 
\kleisli{T_1}{\delta}{\pi}{\gamma};\kleisli{T_2}{\delta}{\gamma}{\rho}\\
& = &
\kleisli{T_1;T_2}{\delta}{\pi}{\rho} 
\end{eqnarray*}

For the case $T=\bct$, likewise the case for $\bot_s$, the definition is purely functional.

For the case $T=\sum_{m\subseteq i}\cop$, let us assume that for any $m\subseteq i$ we have
$\kleisli{T_m}{\delta}{\epsilon}{\rho} = 
\chi^{\epsilon}_{\epsilon};\kleisli{T_m}{\delta}{\epsilon}{\rho'};\chi^{\rho'}_{\rho} = 
\kleisli{T_m}{\delta}{\epsilon}{\rho'};\chi^{\rho'}_{\rho}$.
Then, we have
\begin{eqnarray*}
\chi^{\pi}_{\pi'};\kleisli{\sum_{m\subseteq i}\cop}{\delta}{\pi'}{\rho'};\chi^{\rho'}_{\rho}
& = &
\chi^{\pi}_{\pi'};[\kleisli{T_{\pi'^{-1}(1)}}{\delta}{\epsilon}{\rho'},..., \kleisli{T_{\pi'^{-1}(2^{|i|})}}{\delta}{\epsilon}{\rho'}];\chi^{\rho'}_{\rho} \\
& = &
\chi^{\pi}_{\pi'};[\kleisli{T_{\pi'^{-1}(1)}}{\delta}{\epsilon}{\rho'};\chi^{\rho'}_{\rho},..., \kleisli{T_{\pi'^{-1}(2^{|i|})}}{\delta}{\epsilon}{\rho'};\chi^{\rho'}_{\rho}] \\
& = &
\chi^{\pi}_{\pi'};[\kleisli{T_{\pi'^{-1}(1)}}{\delta}{\epsilon}{\rho},..., \kleisli{T_{\pi'^{-1}(2^{|i|})}}{\delta}{\epsilon}{\rho}] \\
& = &
[\kleisli{T_{\pi^{-1}(1)}}{\delta}{\epsilon}{\rho},..., \kleisli{T_{\pi^{-1}(2^{|i|})}}{\delta}{\epsilon}{\rho}] \\
& = &
\kleisli{\sum_{m\subseteq i}\cop}{\delta}{\pi}{\rho}
\end{eqnarray*}
\end{proof}

\begin{proposition}
The definition of $\kleisli{T}{\delta}{\pi}{\rho}$ is well given.
\end{proposition}
\begin{proof}
We must show that: (1)~the typing is consistent with the definition, (2)~that the choice of $\pi_1,\rho_1,\pi_2,\rho_2$ in Eq.~\eqref{eq:par} and of $\gamma$ in Eq.~\eqref{eq:seq} is inessential for the result, and (3)~that $\kleisli{T_1\oplus T_2}{\delta}{\pi}{\rho} = \kleisli{T_2\oplus T_1}{\delta}{\pi}{\rho}$.

For (1), we must prove that if  $T:\termtype{i}{m}{s}{o}$, $\pi$ is a string over $i$ and $\rho$ is a string over $o$, then $\kleisli{T}{\delta}{\pi}{\rho}:2^{|i|}\to 2^{|o|}$.
The proof is a straightforward rule induction.

For (2), we just exploit Proposition~\ref{prop:order}. 
In the case of Eq.~\eqref{eq:par}, we have
\begin{eqnarray*}
\kleisli{T_1 \oplus T_2}{\delta}{\pi}{\rho} 
& = & 
\chi^{\pi}_{\pi_1\pi_2};(\kleisli{T_1}{\delta}{\pi_1}{\rho_1} \otimes \kleisli{T_2}{\delta}{\pi_2}{\rho_2});\chi^{\rho_1\rho_2}_{\rho} \\
& = & 
\chi^{\pi}_{\pi_1\pi_2};( (\chi^{\pi_1}_{\pi'_1};\kleisli{T_1}{\delta}{\pi'_1}{\rho'_1};\chi^{\rho'_1}_{\rho_1}) \otimes (\chi^{\pi_2}_{\pi'_2};\kleisli{T_2}{\delta}{\pi'_2}{\rho'_2};\chi^{\rho'_2}_{\rho_2}) );\chi^{\rho_1\rho_2}_{\rho} \\
& = & 
\chi^{\pi}_{\pi_1\pi_2};(\chi^{\pi_1}_{\pi'_1}\otimes \chi^{\pi_2}_{\pi'_2}) ; (\kleisli{T_1}{\delta}{\pi'_1}{\rho'_1} \otimes \kleisli{T_2}{\delta}{\pi'_2}{\rho'_2}) ; (\chi^{\rho'_1}_{\rho_1} \otimes \chi^{\rho'_2}_{\rho_2});\chi^{\rho_1\rho_2}_{\rho} \\
& = & 
\chi^{\pi}_{\pi_1\pi_2} ; \chi^{\pi_1\pi_2}_{\pi'_1\pi'_2}; (\kleisli{T_1}{\delta}{\pi'_1}{\rho'_1} \otimes \kleisli{T_2}{\delta}{\pi'_2}{\rho'_2}) ; \chi^{\rho'_1\rho'_2}_{\rho_1\rho_2} ; \chi^{\rho_1\rho_2}_{\rho} \\
& = & 
\chi^{\pi}_{\pi'_1\pi'_2} ; (\kleisli{T_1}{\delta}{\pi'_1}{\rho'_1} \otimes \kleisli{T_2}{\delta}{\pi'_2}{\rho'_2}) ; \chi^{\rho'_1\rho'_2}_{\rho} 
\end{eqnarray*}

In the case of Eq.~\eqref{eq:seq}, we have
\begin{eqnarray*}
\kleisli{T_1 ; T_2}{\delta}{\pi}{\rho} 
& = &
\kleisli{T_1}{\delta}{\pi}{\gamma};\kleisli{T_2}{\delta}{\gamma}{\rho} \\
& = &
\kleisli{T_1}{\delta}{\pi}{\gamma'};\chi^{\gamma'}_{\gamma};\chi^{\gamma}_{\gamma'};\kleisli{T_2}{\delta}{\gamma'}{\rho} \\
& = &
\kleisli{T_1}{\delta}{\pi}{\gamma'};\kleisli{T_2}{\delta}{\gamma'}{\rho}
\end{eqnarray*}

Finally, for~(3), we have:
\begin{eqnarray*}
\kleisli{T_1 \oplus T_2}{\delta}{\pi}{\rho} 
& = & 
\chi^{\pi}_{\pi_1\pi_2};(\kleisli{T_1}{\delta}{\pi_1}{\rho_1} \otimes \kleisli{T_2}{\delta}{\pi_2}{\rho_2});\chi^{\rho_1\rho_2}_{\rho} \\
& = & 
\chi^{\pi}_{\pi_1\pi_2};\chi^{\pi_1\pi_2}_{\pi_2\pi_1};(\kleisli{T_2}{\delta}{\pi_2}{\rho_2} \otimes \kleisli{T_1}{\delta}{\pi_1}{\rho_1});\chi^{\rho_2\rho_1}_{\rho_1\rho_2};\chi^{\rho_1\rho_2}_{\rho} \\
& = & 
\chi^{\pi}_{\pi_2\pi_1};(\kleisli{T_2}{\delta}{\pi_2}{\rho_2} \otimes \kleisli{T_1}{\delta}{\pi_1}{\rho_1});\chi^{\rho_2\rho_1}_{\rho} \\
& = & 
\kleisli{T_2 \oplus T_1}{\delta}{\pi}{\rho} 
\end{eqnarray*}
\end{proof}

% !TEX root =  main.tex

%EXAMPLES NEED TO BE REVISITED}
\begin{example}

Consider the net depicted in Fig.~\ref{fig:N} and the corresponding term calculated in Example~\ref{ex:to-terms-N}.
We show the encoding of the net by considering a generic distribution $\delta$ and use lexicographic order of places. 
We start from Eq.~\ref{eq:terms-N}.
$$
\termsof{N,\{\removedp2, \removedp3\}} = \termsof{N_{\bc_1}}\oplus \termsof{N_{\bc_2},\{\removedp2\}}; \termsof{N_{\bc_3},\{\removedp3\}} \oplus \termsof{I_{\{\removedp5\}}}
$$

\begin{figure}[tp]
%\begin{subfigure}[String diagram for $N$\label{fig:string-N}]{
          $$
          \xymatrix@R=1pc@C=-2pc{
            &&\ar@{-}^(.6){1}[dd]
            && &
            \\
            \\
            &&{\kleisli{\termsof{N_{\bc_1}}}{\delta}{}{}}\ar@{-}^(.7){4}@<-2ex>[dd]
            \ar@{-}@<+2ex>`d[dr] 
             	`[rrrdd]
                 `[rrrddd] 
                 `_l[dddl]
            	[ddddl]<0ex> ^5
              \POS[]+<0pc,0pc> *+=<5pc,1.3pc>[F-]{}
           & 
           & 
            {\kleisli{\termsof{N_{\bc_2},\{\removedp2\}}}{\delta}{}{}}\ar@{-}^(.7){6}[dd]   
                 \POS[]+<0pc,0pc> *+=<6pc,1.3pc>[F-]{}
           \\
           &{\phantom{\termsof{N_{\bc_2}}}}&{\phantom{\termsof{N_{\bc_2}}}}&&&&
           \\
           &{\phantom{\termsof{N_{\bc_2}}}}&
            {\phantom{\termsof{N_{\bc_2}}}}
            &
            {\kleisli{\termsof{N_{\bc_3},\{\removedp3\}}}{\delta}{}{}}
             \POS[]+<0pc,0pc> *+=<10pc,1.3pc>[F-]{}
	     \ar@{-}^(.8){7}@<-6ex>[dd]\ar@{-}^(.8){8}@<-2ex>[dd]\ar@{-}^(.8){9}@<2ex>[dd]\ar@{-}^(.8){10}@<7ex>[dd]
            &
            {\phantom{\termsof{N_{\bc_2}}}}
            &
	   \\            {\phantom{\termsof{N_{\bc_2}}}}
	   &&
	   {\phantom{\termsof{N_{\bc_2}}}}
	   &&&&
	   \\
	  &
          &{\phantom{\termsof{N_{\bc_2}}}}&
          &
          &
          }
          $$
          %}
%\end{subfigure}
\caption{String diagram for $\kleisli{\termsof{N,\{\removedp2, \removedp3\}}}{\delta}{}{}$}\label{fig:string-N}
 \end{figure}

Then, the string diagram for  $\kleisli{\termsof{N,\{\removedp2, \removedp3\}}}{\delta}{\removedp1}{\removedp5,\removedp7,\removedp8,\removedp9,\removedp10}$ is
shown in Fig.~\ref{fig:string-N} and can be computed as follows. 
$$
\begin{array}{cll}
& \kleisli{\termsof{N,\{\removedp2, \removedp3\}}}{\delta}{\removedp1}{\removedp5,\removedp7,\removedp8,\removedp9,\removedp10} 
\\
= 
&
\kleisli{\termsof{N_{\bc_1}}\oplus \termsof{N_{\bc_2},\{\removedp2\}}; \termsof{N_{\bc_3},\{\removedp3\}} \oplus \termsof{I_{\{\removedp5\}}}
}{\delta}{\removedp1}{\removedp5,\removedp7,\removedp8,\removedp9,\removedp10}
&
\mbox{by def.}
\\
=
&
\kleisli{\termsof{N_{\bc_1}}\oplus \termsof{N_{\bc_2},\{\removedp2\}}}{\delta}{\removedp1}{\removedp4,\removedp5,\removedp6}
; 
\kleisli{\termsof{N_{\bc_3},\{\removedp3\}} \oplus \termsof{I_{\{\removedp5\}}}
}{\delta}{\removedp4,\removedp5,\removedp6}{\removedp5,\removedp7,\removedp8,\removedp9,\removedp10}
&
\mbox{by}\ \eqref{eq:seq}\\
=
&
\chi^{\removedp1}_{\removedp1\epsilon};\kleisli{\termsof{N_{\bc_1}}}{\delta}{\removedp1}{\removedp4,\removedp5}
\otimes
\kleisli{\termsof{N_{\bc_2},\{\removedp2\}}}{\delta}{\epsilon}{\removedp6};\chi^{\removedp4,\removedp5,\removedp6}_{\removedp4,\removedp5,\removedp6}
; 
&
\mbox{by}\ \eqref{eq:par}
\\
&
\chi^{\removedp4,\removedp5,\removedp6}_{\removedp4,\removedp6,\removedp5};
\kleisli{\termsof{N_{\bc_3},\{\removedp3\}}}{\delta}{\removedp4,\removedp6}{\removedp7,\removedp8,\removedp9,\removedp10}
\otimes
\kleisli{\termsof{I_{\{\removedp5\}}}}{\delta}{\removedp5}{\removedp5};\chi^{\removedp7,\removedp8,\removedp9,\removedp10,\removedp5}_{\removedp5,\removedp7,\removedp8,\removedp9,\removedp10}
\end{array}
$$

We now show the calculation for each of the boxes in Fig.~\ref{fig:string-N}. 
To ease readability, in the following we let
\[
\begin{array}{rclcrcl}
\cst_{a} & = & {\bct[\{\removedp1\}][\{\removedp4,\removedp5\}][\{\{a\},\{b\}\}]} & \quad &
\cst_{c} & = & {\bct[\{\removedp2\}][\{\removedp6\}][\{\{c\},\{d\}\}]} \\
\cst_{e} & = & {\bct[\{\removedp3\}][\{\removedp7\}][\{\{e\}\}]}  & \quad &
\cst_{g} & = & \bct[{}][\{\removedp9,\removedp{10}\}][{\{\{g\},\{h\}\}}]) \\
\cst_{f} & = & {\bct[{}][\{\removedp{8},\removedp{9},\removedp9,\removedp{10}\}][{\{\{f\},\{e,g\},\{e,h\}\}}]}
\end{array}
\]

For
$\kleisli{\termsof{N_{\bc_1}}}{\delta}{\removedp1}{\removedp4,\removedp5}$, we start from Eq.~\eqref{eq:terms-bc1}, i.e.,
$$  \termsof{N_{\bc_1}} =   \cop[\emptyset][\bot_{\{\removedp4, \removedp5\}}]
  +\cop[\{\removedp1\}][\cst_{a}]
$$
By Eq.~\eqref{eq:sum},
\begin{equation}
\begin{array}{cclcc|c|c|c|c|c|l}
\cline{6-10}
%&&&&&&\co{\removedp3}\co{\removedp4} & {\removedp3}\co{\removedp4} & \co{\removedp3}{\removedp4} & \removedp3\removedp4
&&&&&&\emptyset & \{\removedp4\} & \{\removedp5\} & \{\removedp4, \removedp5\}
\\
\cline{6-10}
 \multirow{2}{*}{$\kleisli{\termsof{N_{\bc_1}}}{\delta}{\removedp1}{\removedp4,\removedp5} = $\!\!\!\!} 
 & 
 \ldelim[{2}{0.1pt}
 & 
 \kleisli{\bot_{\{\removedp4, \removedp5\}}}{\delta}{\epsilon}{\removedp4,\removedp5}\!\!\!\!\!
 & 
 \rdelim]{2}{.1pt} 
 &
 \multirow{2}{*}{$ = $\ }
 &\emptyset & 1 & 0 & 0 & 0 
\\
\cline{6-10}
 &
 &
 \kleisli{\cst_{a}}{\delta}{\epsilon}{\removedp4,\removedp5}\!\!\!\!\!
 &
 &
 &\{\removedp1\} & 0 & p_a & 1-p_a & 0  
\\
\cline{6-10}
\end{array}
\label{eq:matrix-C1}
\end{equation}
where the first row in the table corresponds to $\delta^{|\{4,5\}|}_0$,  as prescribed by Eq.~\eqref{eq:bot}. The 
second row is obtained by Eq.~\eqref{eq:bct}, by assuming that $\delta_{\cst_{a}}(\{a\}) = p_a$
and $\delta_{\cst_{a}}(\{b\}) = 1 - p_a$.

For $\kleisli{\termsof{N_{\bc_2},\{\removedp2\}}}{\delta}{\epsilon}{\removedp6}$, we start from  Eq.~\eqref{eq:terms-bc2}, i.e., 
$$
  \termsof{N_{\bc_2},\{\removedp2\}} = \cst_{c}
$$
Then,
\begin{equation}
\begin{array}{cclcc|c|c|c|c|c|l}
\cline{6-8}
%&&&&&&\co{\removedp3}\co{\removedp4} & {\removedp3}\co{\removedp4} & \co{\removedp3}{\removedp4} & \removedp3\removedp4
&&&&&&\emptyset & \{\removedp6\} 
\\
\cline{6-8}
 \kleisli{\termsof{N_{\bc_2},\{\removedp2\}}}{\delta}{\epsilon}{\removedp6} = \!\!\!\!
 & 
 & 
 \kleisli{\cst_{c}}{\delta}{\epsilon}{\removedp6}\!\!\!\!\!
 & 
 &
  = \ 
 &\emptyset & 1 - p_c & p_c
\\
\cline{6-8}
\end{array}
\label{eq:matrix-C2}
\end{equation}
where $\delta_{\cst_{c}}(\{c\}) = p_c$
and $\delta_{\cst_{c}}(\{d\}) = 1 - p_c$.

For $\kleisli{\termsof{N_{\bc_3},\{\removedp3\}}}{\delta}{\removedp4,\removedp6}{\removedp7,\removedp8,\removedp9,\removedp10}$, we start from 
Eq.~\eqref{eq:terms-bc3}, i.e., 
$$
\begin{array}{l@{\ }c@{\ }lll}
\termsof{N_{\bc_3},\{\removedp3\}} = &&
    \cop[\emptyset][(\bot_{\{\removedp8,\removedp9,\removedp{10}\}} \oplus {\cst_{e}})]
    \\
  &+& \cop[\{\removedp4\}][{(\bot_{\{\removedp8,\removedp9,\removedp{10}\}} \oplus {\cst_{e}})}]
    \\
  &+&
  \cop[\{\removedp6\}][{(\bot_{\{\removedp8\}} \oplus  \cst_{e} \oplus \cst_{g})}]
    \\
  &+&\cop[\{\removedp4,\removedp6\}][{\cst_{f}}]\\
  \end{array}
$$

\begin{equation}
\begin{array}{l@{\ = \ } l}
 \kleisli{\termsof{N_{\bc_3},\{\removedp3\}}}{\delta}{\epsilon}{\removedp7,\removedp8,\removedp9,\removedp10} 
 &
\begin{array}{clc}
 \ldelim[{4}{0.1pt}
 & 
 \kleisli{\bot_{\{\removedp8,\removedp9,\removedp{10}\}} \oplus {\cst_{e}}}{\delta}{\epsilon}{\removedp7,\removedp8,\removedp9,\removedp10} 
 & 
\rdelim]{4}{.1pt} 
\\
 &
 \kleisli{{\bot_{\{\removedp8,\removedp9,\removedp{10}\}} \oplus {\cst_{e}}}}{\delta}{\epsilon}{\removedp7,\removedp8,\removedp9,\removedp10}
\\
&
\kleisli{\bot_{\{\removedp8\}} \oplus  \cst_{e} \oplus \cst_{g}}{\delta}{\epsilon}{\removedp7,\removedp8,\removedp9,\removedp10}
\\
&
\kleisli{\cst_{f}}{\delta}{\epsilon}{\removedp7,\removedp8,\removedp9,\removedp10}
\end{array}
\\[20pt]
\multicolumn{2}{c}{}\\
&\quad
\begin{array}{|c|c|c|c|c|c|c|}
\hline
& \emptyset & \{\removedp7\} & \{\removedp7,\removedp9\} & \{\removedp7,\removedp{10}\} & \{\removedp8\} & \ldots 
\\
\hline
\emptyset & 0 & 1 & 0 & 0 & 0 & 0
\\
\hline
\{\removedp4\} & 0 & 1 & 0 & 0 & 0 & 0
\\
\hline
\{\removedp6\} & 0 & 0 & p_g & 1 - p_g & 0 & 0
\\
\hline
\{\removedp4,\removedp6\} & 0 & 0 & p_g' & 1-p_f - p_g' & p_f & 0
\\
\hline
\end{array}
\end{array}
\label{eq:matrix-C3}
\end{equation}
where the last column (i.e., the one tagged with dots) represents all the remaining nine (inessential) cases. 
The first two rows are obtained as follows: 
$$
\begin{array}{l@{\ =\ }l}
\kleisli{\bot_{\{\removedp8,\removedp9,\removedp{10}\}} \oplus {\cst_{e}}}{\delta}{\epsilon}{\removedp7,\removedp8,\removedp9,\removedp10} 
&
\kleisli{\bot_{\{\removedp8,\removedp9,\removedp{10}\}}}{\delta}{\epsilon}{\removedp8,\removedp9,\removedp10}
\otimes
\kleisli{\cst_{e}}{\delta}{\epsilon}{\removedp7};
\chi^{\removedp8,\removedp9,\removedp10,\removedp7}_{\removedp7,\removedp8,\removedp9,\removedp10}
\\[20pt]
&
\begin{array}{|c|c|c|c|c|c|c|}
\hline
& \emptyset & \ldots
\\
\hline
\emptyset & 1 & 0 
\\
\hline
\end{array}
\otimes
\begin{array}{|c|c|c|c|c|c|c|}
\hline
& \emptyset & \{\removedp7\}
\\
\hline
\emptyset & 0 & 1 
\\
\hline
\end{array}
\\[20pt]
&
\begin{array}{|c|c|c|c|c|c|c|}
\hline
& \emptyset & \{\removedp7\} & \ldots
\\
\hline
\emptyset & 0 & 1 & 0
\\
\hline
\end{array}
\end{array}
$$

The third row is obtained analogously after fixing  $\delta_{\cst_{g}}(\{g\}) = p_g$
and $\delta_{\cst_{g}}(\{h\}) = 1 - p_g$. The last row is obtained by 
Eq.~\eqref{eq:sum} and taking 
$\delta_{\cst_{f}}(\{f\}) = p_f$, 
$\delta_{\cst_{f}}(\{e,g\}) = p_g'$, 
and 
$\delta_{\cst_{f}}(\{e, h\}) = 1 - p_f - p_g'$.

\end{example}

\section{Forward and Backward Inference and Disintegration}\label{sec:inference}
 
In this section we illustrate how to perform bayesian reasoning over Petri nets by 
following the approach presented in~\cite{DBLP:journals/corr/abs-1709-00322}. We first recall some notions, which will be 
used in our reasoning. {\em Marginalisation} is an operation $\Pi_1:X\oplus Y\to X$
 that projects a joint distribution 
$P(x,y)$ on $X\oplus Y$ to the marginal distribution on $X$ computed as 
$P(x) = \sum_y P(x,y)$. Similarly, we  have
$\Pi_2:X\oplus Y\to Y$ for the  projection of  $P(x,y)$ 
over $Y$ defined as $P(y) = \sum_y P(x,y)$.

Consider the arrow $\termsof{N,\{\removedp2, \removedp3\}}:2^1\to2^5$ in Fig.~\ref{fig:string-N}
and suppose we are interested in reasoning about the probability of producing a token  
in the place $\removedp7$. In such case,  marginalisation can be used to obtain an 
arrow $f :2^1\to2^1$ that discards the wires
corresponding to the places $\removedp5$, $\removedp8$, $\removedp9$ and $\removedp{10}$, 
as shown in Fig.~\ref{fig:simplified-string-N}.
\begin{figure}[tp]
%\begin{subfigure}[String diagram for $N$\label{fig:string-N}]{
          $$
          \xymatrix@R=.5pc@C=-2pc{
            &&\ar@{-}^(.6){1}[dd]
            && &
            \\
            \\
            &&{\kleisli{\termsof{N_{\bc_1}}}{\delta}{}{}}\ar@{-}^(.7){4}@<-2ex>[dd]
            \ar@{-}^(.7){5}@{-|}@<+2ex>[d]
              \POS[]+<0pc,0pc> *+=<5pc,1.3pc>[F-]{}
           & 
           & 
            {\kleisli{\termsof{N_{\bc_2},\{\removedp2\}}}{\delta}{}{}}\ar@{-}^(.7){6}[dd]   
                 \POS[]+<0pc,0pc> *+=<6pc,1.3pc>[F-]{}
           \\
           &{\phantom{\termsof{N_{\bc_2}}}}&{\phantom{\termsof{N_{\bc_2}}}}&&&&
           \\
           &{\phantom{\termsof{N_{\bc_2}}}}&
            {\phantom{\termsof{N_{\bc_2}}}}
            &
            {\kleisli{\termsof{N_{\bc_3},\{\removedp3\}}}{\delta}{}{}}
             \POS[]+<0pc,0pc> *+=<10pc,1.3pc>[F-]{}
	     \ar@{-}^(.8){7}@<-6ex>[d]\ar@{-|}^(.8){8}@<-2ex>[d]\ar@{-|}^(.8){9}@<2ex>[d]\ar@{-|}^(.8){10}@<7ex>[d]
            &
            {\phantom{\termsof{N_{\bc_2}}}}
            &
	   \\            {\phantom{\termsof{N_{\bc_2}}}}
	   &&
	   {\phantom{\termsof{N_{\bc_2}}}}
	   &&&&
          }
          $$
          %}
%\end{subfigure}
\caption{Simplified string diagram for $\kleisli{\termsof{N,\{\removedp2, \removedp3\}}}{\delta}{}{}$}\label{fig:simplified-string-N}
 \end{figure}
The wire diagram corresponds to the term:
$$
(\kleisli{\termsof{N_{\bc_1}}}{\delta}{\removedp1}{\removedp4,\removedp5}; \Pi_1) 
\otimes 
\kleisli{\termsof{N_{\bc_2},\{\removedp2\}}}{\delta}{\epsilon}{\removedp6};
(\kleisli{\termsof{N_{\bc_3},\{\removedp3\}}}{\delta}{\removedp4,\removedp6}{\removedp7,\removedp8,\removedp9,\removedp{10}}; \Pi_1 \otimes \Pi_1; \Pi_1)
$$
From Eq.~\eqref{eq:matrix-C1}, we  obtain
\begin{equation}
\begin{array}{c|c|c|c|c|}
\cline{2-4}
%&&&&&&\co{\removedp3}\co{\removedp4} & {\removedp3}\co{\removedp4} & \co{\removedp3}{\removedp4} & \removedp3\removedp4
&&\emptyset & \{\removedp4\} 
\\
\cline{2-4}
 \multirow{2}{*}{$\alpha = \kleisli{\termsof{N_{\bc_1}}}{\delta}{\removedp1}{\removedp4,\removedp5}; \Pi_1= $} 
&\emptyset & 1 & 0  
\\
\cline{2-4}
 &\{\removedp1\} & 1-p_a & p_a   
\\
\cline{2-4}
\end{array}
\end{equation}
Analogously, from Eq.~\eqref{eq:matrix-C3}
\begin{equation}
\begin{array}{l@{\ = \ } l}
 \gamma = \kleisli{\termsof{N_{\bc_3},\{\removedp3\}}}{\delta}{\epsilon}{\removedp7,\removedp8,\removedp9,\removedp10}; \Pi_1\otimes\Pi_1;\Pi_1
&
\begin{array}{|c|c|c|c|c|c|c|}
\hline
& \emptyset & \{\removedp7\}  
\\
\hline
\emptyset & 0 & 1 
\\
\hline
\{\removedp4\} & 0 & 1 
\\
\hline
\{\removedp6\} & 0 & 1 
\\
\hline
\{\removedp4,\removedp6\} & p_f & 1-p_f 
\\
\hline
\end{array}
\end{array}
\end{equation}
We write $\beta$ for $\kleisli{\termsof{N_{\bc_2},\{\removedp2\}}}{\delta}{\epsilon}{\removedp6}$ in Eq.~\ref{eq:matrix-C2}.

Then, $\alpha\otimes \beta$ is obtained as 
\begin{equation}
\begin{array}{ccc|c|c|c|c|c|l}
\cline{4-8}
%&&&&&&\co{\removedp3}\co{\removedp4} & {\removedp3}\co{\removedp4} & \co{\removedp3}{\removedp4} & \removedp3\removedp4
&&&&\emptyset & \{\removedp4\} & \{\removedp6\} & \{\removedp4, \removedp6\}
\\
\cline{4-8}
&
\multirow{2}{*}{$\alpha \otimes \beta =$\!\!\!}
&
&\emptyset & 1-pc & 0 & p_c & 0 
\\
\cline{4-8}
 &
 & 
 &\{\removedp1\} & (1-p_a)(1-p_c) & p_a(1-p_c)  & (1-p_a)p_c & p_ap_c  
\\
\cline{4-8}
\end{array}
\end{equation}
Finally, 
\begin{equation}
\begin{array}{c|c|c|c|c|c|l}
\cline{2-4}
%&&&&&&\co{\removedp3}\co{\removedp4} & {\removedp3}\co{\removedp4} & \co{\removedp3}{\removedp4} & \removedp3\removedp4
&&\emptyset & \{\removedp7\} 
\\
\cline{2-4}
\multirow{2}{*}{$\psi = \alpha \otimes \beta;\gamma =$}
&\emptyset & 0 & 1 
\\
\cline{2-4}
 &\{\removedp1\} & p_ap_cp_f & 1-p_ap_cp_f  
\\
\cline{2-4}
\end{array}
\label{eq:matrix-ABC}
\end{equation}

This means that, given that a token appears in place $\removedp1$ with probability $1$, the place $\removedp7$ will be marked with probability $1-p_ap_cp_f$.
Using the notation in~\cite{DBLP:journals/entcs/JacobsZ16}, this value is computed by precomposing the state $\omega = \ket{1}{\{\removedp1\}}$ with the arrow $\psi$, i.e., by letting $\psi_*(\omega)=\omega;\psi = \ket{p_ap_cp_f}{\emptyset} + \ket{(1-p_ap_cp_f)}{\{\removedp7\}}$.

As an example of backward reasoning, given the \emph{a priori} probability $\frac{1}{2}$ that a token can appear in place $\removedp1$, we can compute the probability that place $\removedp1$ is marked given that a token appears in place $\removedp7$, which is
$$
\frac{1-p_ap_cp_f }{1+(1-p_ap_cp_f)} = \frac{1-p_ap_cp_f }{2-p_ap_cp_f}
$$ 
Using the notation in~\cite{DBLP:journals/entcs/JacobsZ16}, 
this value is computed by setting (for $\psi:X\to \D(Y)$ and $q$ a predicate on $Y$)
\begin{eqnarray*}
\psi^*(q)(x) 
& = &
\sum_{y\in Y} \psi(x)(y)\cdot q(y)
\\
& = & 
\psi(x)(\emptyset)\cdot q(\emptyset) + \psi(x)(\{\removedp7\})\cdot q(\{\removedp7\})
\\
& = & 
\psi(x)(\{\removedp7\})
\end{eqnarray*}
where $q$ is the predicate such that $q(\{\removedp7\})=1$ (and $q(\emptyset)=0$) and then computing
\begin{eqnarray*}
\omega_{|\psi^*(q)} 
& = & 
\sum_{x\in X} \ket{\frac{\omega(x)\cdot\psi^*(q)(x)}{\omega\models \psi^*(q)}}{x}
\\
& = & 
\ket{\frac{\omega(\emptyset)\cdot\psi^*(q)(\emptyset)}{\omega\models \psi^*(q)}}{\emptyset}
+
\ket{\frac{\omega(\{\removedp1\})\cdot\psi^*(q)(\{\removedp1\})}{\omega\models \psi^*(q)}}{\{\removedp1\}}
\\
& = & 
\ket{\frac{\frac{1}{2}\cdot 1}{\omega\models \psi^*(q)}}{\emptyset}
+
\ket{\frac{\frac{1}{2}\cdot (1-p_ap_cp_f)}{\omega\models \psi^*(q)}}{\{\removedp1\}}
\\
& = & 
\ket{\frac{\frac{1}{2}}{\omega\models \psi^*(q)}}{\emptyset}
+
\ket{\frac{\frac{1-p_ap_cp_f}{2}}{\omega\models \psi^*(q)}}{\{\removedp1\}}
\end{eqnarray*}
where
\begin{eqnarray*}
\omega\models \psi^*(q) 
& = & \sum_{x\in X} \omega(x)\cdot \psi^*(q)(x)
\\
& = & \omega(\emptyset)\cdot \psi^*(q)(\emptyset) + 
\omega(\{\removedp1\})\cdot \psi^*(q)(\{\removedp1\})  
\\
& = &
\frac{1}{2}\cdot 1 + \frac{1}{2}\cdot(1-p_ap_cp_f)
\\
& = &
\frac{2-p_ap_cp_f}{2}
\end{eqnarray*}

%\input{example}

% !TEX root =  main.tex

\section{Conclusion}\label{sec:conc}

In this paper we have shown how to derive a Bayesian network from a probabilistic Petri net in the style of~\cite{DBLP:journals/iandc/AbbesB06,BMM18}. The construction is computed via an intermediate representation of a PN as a term in a monoidal (pre-)category structure, exploiting the string diagram representation of BN outlined in~\cite{DBLP:journals/entcs/JacobsZ16}. As shown in Section~\ref{sec:inference}, the BN representation can then be exploited to reason about conditional probabilities of marking reachability, via forward and backward inference. Notably, when transitions have non-empty post-sets then each marking corresponds to a unique deterministic process (i.e., a unique configuration of the underlying event structure) and thus the inference can be transferred to processes as well.

There are many ways in which PN have been enriched with probabilistic behaviour~\cite{DBLP:conf/performance/DuganTGN84,DBLP:journals/tocs/MarsanCB84,Molloy:1985:DTS:4101.4110,DBLP:conf/apn/EisentrautHK013,kudlek2005probability,haar2002probabilistic,bouillard2009critical,katoen1993modeling}. 
To avoid confusion, most of them replace nondeterminism with probability only in part, or focus on interleaved computations, or introduce time dependent stochastic distributions. The approach considered here differs from the others in the literature because: (1) it is purely probabilistic, (2) it deals well with concurrent computations, (3) it addresses confusion.

In the literature, there are very few papers investigating the connections between PN and BN.
In~\cite{lautenbach:etal:2:2006} the relation is drawn in the opposite direction, i.e., PN are used to encode the reasoning of BN. 
The connection established in this paper provides two views for the same
model: on the one side, the standard token game of the PN view (suitable extended with
probabilistic choices) gives a concrete, probabilistic computational model. On the
other side, the BN semantics allows us to reason about the properties of the computations
of the underlying concrete model.

\section*{References}

\bibliography{biblio}

\appendix

% !TEX root =  main.tex

\section{Correctness of mapping to terms}\label{sec:appe}

 The remaining of this section is devoted to establish a correspondence between 
 the semantics of Abbes and Benveniste for a marked net $(N, m)$ and 
  the corresponding term $\termsof{N,m}$. 
 
 % !TEX root =  main.tex

% !TEX root =  main.tex

\subsection{Prime Event Structures}
\label{sec:eventstruct}
A \emph{prime event structure} (also \emph{PES})
~\cite{DBLP:journals/tcs/NielsenPW81,DBLP:conf/ac/Winskel86} is a triple
$\mathcal{E}=(E,\preceq,\#)$ where: $E$ is the set of \emph{events};
the \emph{causality
  relation} $\preceq$ is a partial order on events; the \emph{conflict relation} $\#$ is a symmetric, irreflexive relation on events
 such that conflicts are inherited by
causality, i.e., $\forall e_1,e_2,e_3\in E.~ e_1\# e_2 \preceq e_3
\Rightarrow e_1 \# e_3$. 

 The %construction of the 
 PES
$\mathcal{E}_N$ associated with a net
$N$ can be formalised using category theory as a chain of universal
constructions, called coreflections.
%, with the category of prime event structures be equivalent to the category of prime algebraic domains. 
Hence, for each PES
$\mathcal{E}$, there is a standard, unique (up to isomorphism) nondeterministic occurrence net $N_{\mathcal{E}}$
that yields $\mathcal{E}$ and thus we can freely move from one setting to
the other.

%\begin{example}
%  Consider the nets in Figs.~\ref{fig:condconfusion} and~\ref{fig:multcondconfusion}. The
%  corresponding PESs  are shown
%  below each net. 
%  Events are in bijective correspondence with the transitions of the nets.  
%  Strict causality is depicted by arrows and immediate conflict by curly lines.  
%\qed
%\end{example}

%The notion of confusion can be extended to event structures.
Given an event $e$, its \emph{downward closure} $\down{e} = \{ e'\in E \mid e'\preceq e\}$ is the set of causes of $e$.
As usual, we assume that $\down{e}$ is finite for any $e$.
Given $B\subseteq E$, we say that $B$ is \emph{downward closed} if
$\forall e\in B.~ \down{e} \subseteq B$ and that $B$ is \emph{conflict-free} if $\forall e,e'\in B.~ \neg (e \# e')$.
We let the \emph{immediate conflict} relation $\#_0$ be
defined on events by letting $e \#_0 e'$ iff $(\down{e} \times \down{e'}) \cap
\# = \{(e,e')\}$, i.e., two events are in immediate conflict if they
are in conflict but their causes are compatible.
%%
%Then, an event structure is \emph{confusion free} if its maximal set of events that are pairwise in immediate conflict and that have the same set of causal predecessors are closed under immediate conflict.

\subsection{Abbes and Benveniste's Branching Cells}\label{sec:ABcells}

%In order to define AB's branching cells, some terminology must be introduced first.
%We start by introducing some terminology. 
In the following we assume that a finite PES $\mathcal{E}=(E,\preceq,\#)$ is given.
A \emph{prefix} $B\subseteq E$ is any downward-closed set of events
(possibly with conflicts). Any prefix $B$  induces an event structure $\mathcal{E}_B
=(B, \preceq_B, \#_B)$ where $\preceq_B$ and $\#_B$ are
the restrictions of $\preceq$ and $\#$ to the events in $B$.
%
%It is immediate to check that any (downward-closed) subset $B\subseteq
%E$ implicitly defines a \mbox{(sub-)}event structure
%$B_\mathcal{E}=(B,\preceq_B,\#_B)$, where $\preceq_B$ and $\#_B$ are
%the obvious restrictions of $\preceq$ and $\#$ to events in $B$.
%%, i.e., $\preceq_B = \preceq \cap (B \times B)$ and $\#_B = \# \cap (B \times B)$. 
%This fact is important, because the definition of the branching
%cells that are active at some point of the computation takes into
%account the sub-event structure without all the events that are in
%conflict with the events executed so far.
%
A \emph{stopping prefix} is a prefix $B$ that is closed under
immediate conflicts, i.e., $\forall e\in B, e'\in E.~ e\#_0 e'
\Rightarrow e'\in B$.  Intuitively, a stopping prefix is a prefix
whose (immediate) choices are all available.  
It is \emph{initial} if the only stopping prefix strictly
included in $B$ is $\emptyset$.  

%We assume that any $e\in
%E$ is contained in a finite stopping prefix, i.e., 
%we consider \emph{local finite event structures}, which can  
%be seen as event structures with bounded confusion~\cite{DBLP:journals/iandc/AbbesB06}.

A \emph{configuration} $v \subseteq \mathcal{E}$ is any set of events
that is downward closed and conflict-free.  Intuitively, a
configuration represents (the state reached after executing) a
concurrent but deterministic computation of $\mathcal{E}$.
Configurations are ordered by inclusion and we denote by
$\mathcal{V}_{\mathcal{E}}$ the poset of
 configurations of $\mathcal{E}$ and by $\Omega_{\mathcal{E}}$ the poset of maximal
configurations of $\mathcal{E}$.  
%For $B \subseteq E$ we write 
%$\mathcal{V}_{B}$ and $\Omega_{B}$ for  
%$\mathcal{V}_{B_\mathcal{E}}$ and $\Omega_{B_\mathcal{E}}$, respectively.

%Suppose the configuration $v$ has been executed.  
The \emph{future} of a configuration $v$, written $E^v$, is the set of events that can be executed after $v$, i.e.,
$E^v = \{e\in E\setminus v\ \mid\ \forall e'\in v. \neg(e \# e')\}$.
We write $\mathcal{E}^v$ for the event structure induced by $E^v$.  
%We assume that any  configuration enables only
%finitely many events, i.e.,  the set of minimal elements in $E^v$
%w.r.t. $\preceq$ is finite  for any $v\in
%\mathcal{V}_{\mathcal{E}}$.

A configuration $v$ is \emph{stopped } if there is a stopping prefix $B$ with $v\in \Omega_B$.
%
%A configuration
and $v$ is \emph{recursively stopped} (or r-stopped) if there is a sequence of
configurations $\emptyset= v_0 \subset \ldots \subset v_n=v$ such that for any $i\in [0,n)$ the set $v_{i+1} \setminus v_i$ is a
  stopped configuration of  $\mathcal{E}^{v_i}$ for
  $v_i$ in $\mathcal{E}$.
%
%We denote by $\mathcal{W}_{\mathcal{E}}$ the set of all finite,
%recursively stopped configurations.  In other terms,
%$\mathcal{W}_{\mathcal{E}}$ is the smallest class of configurations
%that contains all stopped configurations and that is closed under
%concatenation~\cite{DBLP:conf/fossacs/AbbesB05}.

A \emph{branching cell} is any initial stopping prefix of the future
$\mathcal{E}^v$ of a  recursively stopped configuration $v$.
%\in\mathcal{W}_{\mathcal{E}}$.  
Intuitively, a branching cell is a
minimal subset of events closed under immediate conflict.  
%We denote by $\mathcal{X}_{\mathcal{E}}$ the set of all branching cells.
%
We remark that branching cells are determined by
considering the whole (future of the) event structure $\mathcal{E}$ and
they are recursively computed as $\mathcal{E}$ is executed. Remarkably, 
every maximal configuration has a branching cell decomposition.

\begin{example}
  \label{ex:bc-decomposition}

  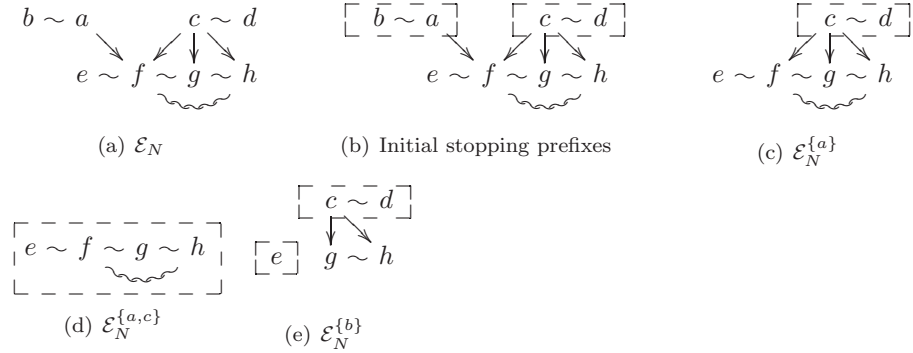
\begin{figure}[t]
    \subfigure[$\mathcal{E}_N$\label{fig:multcondconfusion}]{
      $$
      \xymatrix@R=.6pc@C=.8pc{
        b\ar@{~}[r]
        &
        a\ar[dr]
        &
        &
        c\ar@{~}[r]\ar[dl]\ar[d]\ar[dr]
        &
        d
        \\
        &
        e\ar@{~}[r]
        &
        f\ar@{~}[r]
        &
        g\ar@{~}[r]  
        &
        h\ar@/^1pc/@{~}[ll]
        \\
        {\ }      }
      $$}
    \subfigure[Initial stopping prefixes\label{fig:sp-1}]{
      $$
      \xymatrix@R=.6pc@C=.8pc{
        b\ar@{~}[r]
        &
        a\ar[dr]
        \POS[]-<1pc,0pc> *+=<3.5pc,1pc>[F--]{}
        &
        &
        c\ar@{~}[r]\ar[dl]\ar[d]\ar[dr]
        &
        d
        \POS[]-<1pc,0pc> *+=<3.5pc,1pc>[F--]{}
        \\
        &
        e\ar@{~}[r]
        &
        f\ar@{~}[r]
        &
        g\ar@{~}[r]  
        &
        h\ar@/^1pc/@{~}[ll]
        \\
        {\ }     }
      $$}
    \subfigure[$\mathcal{E}_N^{\{a\}}$\label{fig:sp-2}]{ 
      $$
      \xymatrix@R=.6pc@C=.8pc{
        &
        &
        c\ar@{~}[r]\ar[dl]\ar[d]\ar[dr]
        &
        d
        \POS[]-<1pc,0pc> *+=<3.5pc,1pc>[F--]{}
        \\
        e\ar@{~}[r]
        &
        f\ar@{~}[r]
        &
        g\ar@{~}[r]
        &
        h\ar@/^1pc/@{~}[ll]
        \\
        {\ }      }
      $$}
    \subfigure[$\mathcal{E}_N^{\{a,c\}}$\label{fig:sp-3}]{
    $$
      \xymatrix@R=.6pc@C=.8pc{
        \\
        e\ar@{~}[r]
        &
        f\ar@{~}[r]
        \POS[]+<1pc,-.3pc> *+=<6.5pc,2.2pc>[F--]{}
        &
        g\ar@{~}[r]
        &
        h \ar@/^1pc/@{~}[ll]
        \\
        {\ }      }
      $$
}      
      \subfigure[$\mathcal{E}_N^{\{b\}}$\label{fig:sp-4}]{
      $$
      \xymatrix@R=.6pc@C=.8pc{
        &
        c\ar@{~}[r]\ar[d]\ar[dr]
        &
        d
        \POS[]-<1pc,0pc> *+=<3.5pc,1pc>[F--]{}
        \\
        e
        \POS[]-<0pc,0pc> *+=<1.3pc,1pc>[F--]{}
        &
        g\ar@{~}[r]
        %\POS[]+<.8pc,0pc> *+=<3pc,1pc>[F--]{}
        &
        h\\
        {\ }      }
      $$}
    \caption{AB's branching cell decomposition (running example)}\label{fig:sp-all}
  \end{figure}
% a -> a
% d -> b
% e -> c
% f -> d
% b -> e
% c -> f
% g -> g
% h aggiunto
  Consider the PES $\mathcal{E}_N$ in
  Fig.~A.\ref{fig:multcondconfusion} and its maximal configuration $v =
  \{a,c,e,g\}$. We show that $v$ is recursively stopped by exhibiting a branching cell decomposition. The initial
  stopping prefixes of $\mathcal{E}_N = \mathcal{E}_N^\emptyset$ are
  shown in Fig.~A.\ref{fig:sp-1}. There are two possibilities for
  choosing $v_1\subseteq v$ and $v_1$ recursively stopped: either $v_1 =
  \{a\}$ or $v_1=\{c\}$. When $v_1 = \{a\}$, the choices for $v_2$ are
  determined by the stopping prefixes of $\mathcal{E}_N^{\{a\}}$
  (see Fig.~A.\ref{fig:sp-2}) and the only
  possibility is $v_2 = \{a, c\}$. From $\mathcal{E}_N^{\{a,c\}}$
  in Fig.~A.\ref{fig:sp-3}, we take $v_3 =
  v$. Note that  $\{a,c,e\}$ is not recursively
  stopped because $\{e\}$ is not maximal
  in the stopping prefix of $\mathcal{E}_N^{\{a,c\}}$ (see Fig.~A.\ref{fig:sp-3}).
Finally, note that the branching cells of $\mathcal{E}_N^{\{a\}}$ (Fig.~A.\ref{fig:sp-2}) and $\mathcal{E}_N^{\{b\}}$ (Fig.~A.\ref{fig:sp-4}) 
correspond to different choices in $\mathcal{E}_N^\emptyset$ and thus have different stopping prefixes.
%
%  The choices made in a branching cell influence the branching cells available next.
%%  
%%  Stopping prefixes and branching cells change dynamically as a
%%  consequence of the choices made in other branching cells. 
%%  
%  Take $v' =
%  \{d,e,b,g\}$ and start its decomposition from
%%   with the stopping prefixes of 
%  $\mathcal{E}_N = \mathcal{E}_N^\emptyset$ (Fig.~\ref{fig:sp-1}). 
%  Fix $v_1' = \{d\}$, and get
%  $\mathcal{E}_N^{\{d\}}$ as in Fig.~\ref{fig:sp-4}. 
%  Observe that $\mathcal{E}_N^{\{a\}}$(Fig.~\ref{fig:sp-2}) and
%  $\mathcal{E}_N^{\{d\}}$(Fig.~\ref{fig:sp-4}) have different stopping prefixes  
%  because reached after different choices in $\mathcal{E}_N^\emptyset$ ($v_1=\{a\}$
%  and $v_1'=\{d\}$).
\end{example}

\subsection{AB's decomposition and terms}

%We now establish the correspondence
%between  s-cell decomposition of  a  net $N$ and  the
%recursively stopped configurations of (the event structure of) $N$
%(i.e., of $\mathcal{E}_N$). 
%
The recursively stopped configurations of a marked net $(N,m)$ characterise 
all the allowed executions of $N$ under the marking $m$. Hence, we formally link the recursively stopped configurations of $\mathcal{E}_{(N,m)}$ 
with the deterministic processes associated with $\termsof{N,m}$. We start by introducing the 
notion of configurations associated to a term.

\begin{definition}  Given a term $T:\termtype{i}{}{s}{o}$ and a marking $m\subseteq i$, the set of configurations of 
\ $T$ under $m$, written $\conf{T,m}$, is  defined inductively as follows.
\[\begin{array}{l@{\ =\ }l}
 \conf{I_s,m} & \{\emptyset\} \\
 \conf{\bot_s,\emptyset} & \{\emptyset\} \\
 \conf{T_1\oplus T_2, m} & \{v_1\cup v_2\ |\ \forall j = 1,2.\ T_j:\termtype{i_j}{\emptyset}{s_j}{o_j} \ \\
 \multicolumn{2}{r}{\hfill\wedge\ v_j\in\conf{T_j,m\cap i_j)}\}}\\
 \conf{T_1; T_2, m} & \{ v_1\cup v_2 \ |\ v_1\in\conf{T_1,m}\ \wedge\ T_2:\termtype{i_2}{\emptyset}{s_2}{o_2}\qquad\\
 \multicolumn{2}{r}{\hfill \wedge\ v_2\in\conf{T_2,\maxp{v_1}\cap i_2}\}}\\
  \conf{\bct,\emptyset} & \Theta\\
 \conf{\sum_{m\subseteq i}\ \cop, m_j} & \conf{T_j, \emptyset} 
 \end{array}
\]
\end{definition}
%
%\begin{lemma} If $T:\termtype{i}{}{s}{o}$ and $m\subseteq i$, then $v\in\conf{T,m}$ is 
%a deterministic process. 
%\end{lemma}
%
%\begin{proof} The proof follows by structural induction on $T$. We analyse the interesting cases:
%\begin{itemize}
% \item $\bct$: it follows because of the definition of $\Theta$.
% \item $T_1\oplus T_2$: it follows by using inductive hypothesis to conclude that $v_i \in \conf{T_i,m_i}$
% are deterministic processes. Moreover, $T_1$ and $T_2$ are disjoint because $T$ is well-typed. Hence, 
%$v_1\cup v_2$ is a deterministic process. 
%\item $T_1;T_2$: it follows as in the previous case. 
% \item $\sum_{m\subseteq i}\ \cop$: it follows by inductive hypothesis.
%\end{itemize}
%\end{proof}

%Given a configuration $v$ of $\mathcal{E}_{(N,m)}$,   $\maxp C = \{ \postS e \ |\ e\ \mbox{is maximal in}\ C \}$.

\begin{proposition} Let $(N,m) : i \to o$ be a finite marked occurrence net and $T = \termsof{N,m}$. Then, 
for $j\subseteq i$,
 $v$ is a maximal r-stopped configuration of  $\mathcal{E}_{(N, m\cup j)}$ iff  $v\in\conf{T,j}$.
\end{proposition}

\begin{proof} The proof follows by structural induction on $T$. 
\begin{itemize}

\item $T = I_s$.  For all $j\subseteq i$, we have  $\conf{I_s,j} = \{\emptyset\}$.  Consequently, $v\in\conf{I_s,j}$ implies  $v = \emptyset$.
Since $\termsof{N,m} = I_s$, $(N, m) = I_s$. Then, $s = i$ and $m = \emptyset$. Therefore, $\mathcal{E}_{(N, m\cup j)} =\emptyset$.
Consequently, $v \in \mathcal{E}_{(N, m\cup j)}$ implies $v = \emptyset$.

\item $T = \bot_s$. It holds trivially because there is no $(N,m)$ such that $\termsof{N,m} = \bot_s$. 

\item $T = T_1\oplus T_2$. Then, $(N,m) = (N_1,m_1)\oplus (N_2,m_2)$, $T_1 = \termsof{N_1,m_1}$
$T_2 = \termsof{N_2,m_2}$. By inductive hypothesis, 
$v_i \in \conf{T_i,j_i}$ iff $v_i$ is an r-stopped configuration of $\mathcal{E}_{(N_i,m_i\cup j_i)}$.
The proof follows by noting that the union of two disjoint r-stopped configurations is an r-stopped configuration.

\item $T = T_1; T_2$. Then, $(N,m) = (N_1,m_1); (N_2,m_2)$, $T_1 = \termsof{N_1,m_1}$
$T_2 = \termsof{N_2,m_2}$. By inductive hypothesis, 
$v_i \in \conf{T_i,j_i}$ iff $v_i$ is an r-stopped configuration of $\mathcal{E}_{(N_i,m_i\cup j_i)}$.
The proof follows by noting that $v_1$ is an r-stopped configuration of 
$\mathcal{E}_{(N,m\cup j)}$ and $v_2$ is an r-stopped configuration of
$\mathcal{E}_{(N,m\cup j)}^{v_1}$. Consequently, 
%\todo{in che senso ``sequential composition''?}
$v=v_1\cup v_2$ is an r-stopped configuration of $\mathcal{E}_{(N,m\cup j)}$.

\item $T = \bct[m][o][\Theta(N_{\bc})]$. Then, $N = N_\bc$
 and $m = \minp\bc$. 
Moreover, $v\in\mathcal{E}_{(\bc,\minp\bc)}$ implies that $v$ is a maximal deterministic process 
of $(\bc,\minp\bc)$, i.e., a transaction. Hence, $v\in\Theta(N_{\bc})$ and $v\in\conf{T,\emptyset}$.

\item $T = \sum_{j\subseteq i}\ \cop[j][\bot_{d_{j}} \oplus T_{j}]$ with $T_j = \termsof{\mathsf{can}(N_{\bc}@m\cup j)}$.
Then, $v \in \conf{T,j}$ iff $v \in \conf{T_j,\emptyset}$.
%\todo{In che senso $v \in T_j$?}
 By inductive hypothesis, $v$ is a maximal r-stopped configuration of 
$\mathcal{E}_{N_{\bc}@m\cup j}$. The proof is completed by noting 
that $\mathcal{E}_{N_{\bc}@m\cup j} = \mathcal{E}_{(N_{\bc},m\cup j)}$.
\end{itemize}
\end{proof}

\end{document}